\documentclass[11pt]{article}

\usepackage{geometry}
\usepackage{graphicx}
\usepackage{epstopdf}
\usepackage{booktabs} 
\usepackage{array} 
\usepackage{paralist} 
\usepackage{verbatim} 
\usepackage{subfigure} 
\usepackage{caption}
\usepackage{fancyhdr, bm} 
\usepackage[nottoc,notlof,notlot]{tocbibind}
\usepackage[titles,subfigure]{tocloft}

\usepackage{amsmath,amsfonts,amsthm,mathrsfs,amssymb,cite}
\usepackage[usenames]{color}
\usepackage{mathtools}
\usepackage{float}

\newtheorem{thm}{Theorem}[section]

\theoremstyle{definition}

\newtheorem{rem}{Remark}[section]
\newtheorem{example}{Example}
\numberwithin{equation}{section}

\title{\bf Two gesture-computing approaches by using electromagnetic waves}

\author{Yukun  Guo\thanks{Department of Mathematics,
 Harbin Institute of Technology, Harbin, P. R. China. Email: {\tt ykguo@hit.edu.cn}},
\and Jingzhi Li\thanks{Department of Mathematics, Southern University of Science and Technology, Shenzhen, P. R. China. Email: {\tt li.jz@sustc.edu.cn}},
\and Hongyu Liu\thanks{Department of Mathematics, Hong Kong Baptist University, Kowloon, Hong Kong SAR, P. R. China.  Email:  {\tt hongyuliu@hkbu.edu.hk}}
\and Xianchao Wang\thanks{Department of Mathematics, Harbin Institute of Technology, Harbin, P. R. China. Email: {\tt xcwang90@gmail.com}}}

\date{} 
\begin{document}
\maketitle

\begin{abstract}

 We are concerned with a novel sensor-based gesture input/instruction technology which enables human beings to interact with computers conveniently. The human being wears an emitter on the finger or holds a digital pen that generates a time harmonic point charge. The inputs/instructions are performed through moving the finger or the digital pen. The computer recognizes the instruction by determining the motion trajectory of the dynamic point charge from the collected electromagnetic field measurement data. The identification process is mathematically modelled as a dynamic inverse source problem for time-dependent Maxwell's equations. From a practical point of view, the point source should be assumed to move in an unknown inhomogeneous background medium, which models the human body and the surroundings. Moreover, a salient feature is that the electromagnetic radiated data are only collected in a limited aperture. For the inverse problem, we develop, from the respectively deterministic and stochastic viewpoints,  a dynamic direct sampling method and a modified particle filter method. Both approaches can effectively recover the motion trajectory. Rigorous theoretical justifications are presented for the mathematical modelling and the proposed recovery methods. Extensive numerical experiments are conducted to illustrate the promising features of the two proposed recognition approaches.

\medskip



\end{abstract}

\section{Introduction}\label{sect:1}

Among various human-computer interaction technologies, most people prefer to interact with computers in a more personal way, e.g. by using voice, touch and gesture rather than a mouse or a keyboard. In particular, people take an interest in gesture-computing technology since it enables people to communicate with the machine and interact more naturally without the help of any mechanical devices. Some existing technologies have been developed using cameras to capture the human body gesture, and then using imaging processing algorithms to interpret the body language in order to understand the instructions or inputs; see \cite{Wiki} and the references therein. We also refer the interested readers to a recent article by Liu et al. \cite{LWY}, where they first attempted to achieve gesture recognition by using inverse scattering techniques. In fact in \cite{LWY}, instead of using of a camera, it is proposed that one uses wave probing to identify the body gesture. However, the proposed method in \cite{LWY} is mainly suitable for static instructions for computers, but not suitable for dynamic text inputs. In order to develop a novel gesture-computing method that is suitable for both instruction and input, Guo et al. \cite{GLLW} propose to use a moving emitter generating an acoustic point wave and then to identify the motion trajectory of the emitter that carries the intended instruction/input. The purpose of the present article is to further develop the idea in \cite{GLLW} to produce much more practical, effective and efficient gesture-computing methods.

Fig.~\ref{fig.1} provides a schematic illustration of the proposed gesture-computing technology. In order to give the specific input or instruction to the computer, a human being wears an emitter on his/her finger or holds a digital pen that generates a time harmonic point charge. The input/instruction is performed through moving the finger or the digital pen. There are sensors installed on the computer that timely collect the electromagnetic wave data generated by the motion of the point charge. Then the computer recognizes the instruction/input by determining the motion trajectory of the dynamic point charge from the collected electromagnetic field data. Mathematically, the motion trajectory identification can be modelled as an inverse problem where one intends to identify a moving emitter from the measurement of the electromagnetic wave fields generated by the emitter.
\begin{figure}
  \centering
   \includegraphics[  width=0.6\textwidth]{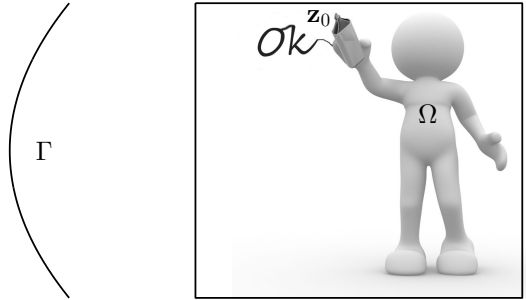}
   \caption{Schematic illustration of the proposed input/instruction technology using a moving emitter. }\label{fig.1}
\end{figure}
There are several practically important issues that should be incorporated into the design and modelling. First, the choice of the point charge is a critical ingredient in our design. The physical principle of using a moving charge to generate electromagnetic fields can be found in \cite{Griffiths}. We choose to utilize a time-harmonic point charge for two considerations. On the one hand, the time-harmonic point charge will generate electromagnetic waves with a fix frequency. This enables us to distinguish the signals due to the motion of the emitter from the possibly various background signals such as those from radios, telephones and televisions. On the other hand, in a certain practical scenario, we can show that by properly choosing the frequency of the waves, one can eliminate the scattering influence from certain inhomogeneous background scatterers such as the body of the human being who performs the input/instruction. Second, the emitter should be modelled as moving within an unknown inhomogeneous background medium as described above, and the trajectory identification should be independent of the background medium. Third, the measurement data should be collected only in a limited aperture on a surface. Indeed, as described earlier, the sensors used to collect the wave data are installed on the computer. Hence, one would only have limited-view data for the inverse problem. Moreover, we shall see in the subsequent study that the location of the measurement surface is also an important ingredient in our design. Finally, the trajectory identification should be conducted in a timely manner. All of those challenging issues distinguish our study from the existing ones in the literature on inverse source problems associated with electromagnetic wave probings. We are aware of some existing theoretical and computational developments on the identification of an unknown source from the measurement of the generated electromagnetic wave data away from the source. Those problems arise in various practical applications; see e.g. \cite{ATW, 2BaoG,DLU,HR,1Mon,10Fokas,Isa2,Yam2} and the references therein. Finally, as mentioned earlier that in \cite{GLLW}, the similar input/instruction technology has been proposed using a moving acoustic emitter. Clearly, using electromagnetic waves for the proposed technology is more realistic and practical. Indeed, according to the mathematical and theoretical analysis in \cite{GLLW}, the applicability of the proposed technology is limited if acoustic wave is used due to the low-speed of propagation. The use of electromagnetic wave shall overcome this problematic issue. Nevertheless, we would like to emphasize that our study is mainly concerned with the conceptual design and theoretical analysis, and the issue on engineering realization of the proposed technology is beyond the scope of this paper.

For the dynamic inverse electromagnetic problem described above, we develop two methods for the trajectory identification: one is a direct imaging method motivated by our theoretical analysis and the other one is a modified particle filter method motived by the particle filter methods developed in \cite{CB,KS,Potthast} for various dynamic inverse problems. Both methods are shown to be effective and efficient by extensive numerical examples. Finally, we also note some existing results on recovering moving targets in different contexts \cite{Amm, GF, 24Naka}.

The rest of the paper is organized as follows. Section 2 is devoted to the mathematical modelling of the moving-emitter-based gesture recognition technique as well as some necessary theoretical analysis. In Section 3, we develop two recovery methods for the trajectory identification. In Section 4, we conduct numerical experiments to illustrate the proposed instruction/input technology as well as the effectiveness and efficiency of the methods developed for the identification.

\section{Mathematical modelling and theoretical analysis}\label{sect:modelling}

In this section, we build up the mathematical modelling and present some theoretical analysis for the proposed sensor-based gesture recognition technology.

\subsection{Mathematical modelling}\label{sect:model}

Assume that the point charge emanates a causal sinusoidal signal
\begin{equation}\label{eq:charge}
q(t):=\begin{cases}
\sin(\omega_0 t), &\quad  t\geq 0,\\
0, &\quad t< 0,
\end{cases}
\end{equation}
where $t$ denotes the  temporal variable, and $\omega_0\in\mathbb{R}_+$ denotes the frequency. Suppose that the electric charge density $\rho$ has the following form
\begin{equation}\label{eq:point wave}
\displaystyle \rho(\bm{x},t)=q(t)\, \delta(\bm{x}-\bm{z}(t)), \quad (\bm x, t)\in\mathbb{R}^3\times\mathbb{R}_+,
\end{equation}
where $\bm{x}$ denotes the spatial variable, and $\bm{z}: \mathbb{R}_{+}\mapsto\mathbb{R}^3$ is a $C^2$ smooth function that signifies the instantaneous position of the point charge at time $t$ and $\delta$ is the Dirac's delta distribution.
 Let the moving trajectory of the point charge be
 \begin{equation}\label{trajectory}
   \Lambda_{\bm{z}}:=\{\bm{z}(t)\mid 0<t\leq T\},
\end{equation}
where $T\in \mathbb{R}_+$ is the terminal time of the motion.
 As the point charge moves, the electric current density satisfies
 \begin{equation}\label{eq:ecd}
  \bm{J}(\bm{x},t)=\rho(\bm{x},t)\bm{v}(t),
 \end{equation}
 where $\bm{v}(t)={\mathrm{d}\bm{z}(t)}/{\mathrm{d}t}$ is the instantaneous  velocity.

Suppose a human being is interacting with the gesture computing device, his/her body could be modelled by a bounded  moving domain   $\widetilde{\Omega}(t)\subset \mathbb{R}^3 (0<t\leq T)$ such that its boundary is time-varying. In what follows, we set
\begin{equation*}
\Omega:=\bigcup\limits_{ t\in(0,T]} \widetilde{\Omega}(t).
\end{equation*}
The receivers are fixed and located on a static open surface $\Gamma \subset \mathbb{R}^3 \backslash \Omega$ such that $\Gamma \cap \Lambda_{\bm{z}}=\emptyset$;
see Figure \ref{fig.1} for the schematic illustration of the problem setting. We set $\epsilon(\bm x, t)\in L^\infty(\mathbb{R}^3\times[0, T])$ and $\mu(\bm x, t)\in L^\infty(\mathbb{R}^3\times [0, T])$ be positive functions to, respectively, signify the electric permittivity and magnetic permeability at the space point $\bm x$ and time point $t$; and $\sigma(\bm x, t)\in L^\infty(\mathbb{R}^3\times [0, T])$ be a non-negative function to signify the electric conductivity at $(\bm x, t)$. We assume that the medium is homogeneous in the background space $\mathbb{R}^3\backslash \widetilde{\Omega}(t)$, that is,
\begin{equation*}
  \epsilon(\bm x,t)=\epsilon_0,\quad \mu(\bm x,t)=\mu_0, \quad   \sigma(\bm x,t)=0,\quad  \mathrm{in}\ \ \  \mathbb{R}^3\backslash \widetilde{\Omega}(t),
\end{equation*}
where $\epsilon_0\in\mathbb{R}_+$ and $\mu_0\in\mathbb{R}_+$ are, respectively, the constant permittivity and permeability of the background space. Biologically, it is reasonable to assume that $\mu(\bm x,t)=\mu_0$ in $\widetilde{\Omega}(t)$.

In the setup described above, the electric current $\bm{J}(\bm x, t)$ generates electromagnetic waves that propagate in the space. We denote by $\bm{E}(\bm{x},t)$  and $\bm{H}(\bm{x},t)$  the electric field and the magnetic field in $ \mathbb{R}^3\times (0,T]$, respectively. The electromagnetic field satisfies the following Maxwell system,
\begin{equation}\label{eq:main1}
  \begin{cases}
  \displaystyle\nabla \times \bm{E} =-\mu_0\,\frac{\partial \bm{H}}{\partial t},\smallskip    \\
  \displaystyle\nabla \times \bm{H} = \bm{J}+\sigma\, \bm{E}+\epsilon\,\frac{\partial \bm{E}}{\partial t}, \smallskip \\
  \end{cases}
\end{equation}
and the initial condition
\begin{equation}\label{eq:initial1}
     \bm{E}|_{t=0} = 0, \quad \bm{H}|_{t=0} = 0.
\end{equation}
The Maxwell system \eqref{eq:main1}--\eqref{eq:initial1} is well understood, and we refer to \cite{Leis} for the study on its well-posedness and especially the unique existence of a pair of solutions $(\bm{E},\bm{H})\in C^1(H_{loc}(\mathrm{curl},\mathbb{R}^3),(0,T])\times C^1(H_{loc}(\mathrm{curl},\mathbb{R}^3),(0,T])$. The trajectory identification associated with the prosed input/instruction technology can be formulated as follows,
\begin{equation}\label{inversion}
\bm{E}\big|_{\Gamma\times(0,T]}\longrightarrow \Lambda_{\bm{z}}.
\end{equation}
That is, by monitoring the change of the electromagnetic field on the surface $\Gamma$ generated by the emitter, we intend to recover the motion trajectory of the emitter. Here, we would like to emphasize that the recovery should be independent of the inhomogeneity $(\Omega;\epsilon,\mu,\sigma)$.

\subsection{A theoretical model approximation}\label{sect:approx}

As emphasized at the end of Section~\ref{sect:model}, the recovery for the inverse problem \eqref{inversion} should be independent of the background inhomogeneity $(\Omega;\epsilon,\mu,\sigma)$. Clearly, in the measurement data $\bm{E}\big|_{\Gamma\times(0,T]}$, there are both scattering data due to the moving emitter and the background inhomogeneity that are coupled together. Furthermore, the background inhomogeneity could be changed due to the change of the human being that performs the input or instruction. Next, we present some practical conditions that the scattering influence due to the background inhomogeneity can be eliminated. To that end, we introduce the following Maxwell system for $(\bm{E}_0,\bm{H}_0)\in C^1(H_{loc}(\mathrm{curl},\mathbb{R}^3),(0,T])\times C^1(H_{loc}(\mathrm{curl},\mathbb{R}^3),(0,T])$,
\begin{equation}\label{eq:MaxwellFree}
\begin{cases}
 \displaystyle\nabla \times \bm{E}_0= -\mu_0\,\frac{\partial \bm{H}_0}{\partial t},\smallskip   \\
  \displaystyle\nabla \times \bm{H}_0 =  \bm{J}+\epsilon_0\,\frac{\partial \bm{E}_0}{\partial t},  \smallskip \\
 \bm{E}_0|_{t=0} = 0, \quad \bm{H}_0|_{t=0} = 0,
\end{cases}
\end{equation}
where $\bm{J}$ is given in \eqref{eq:ecd}. The system \eqref{eq:MaxwellFree} describes the electromagnetic wave propagation generated by the moving emitter in the free space without any inhomogeneity presented. In the following, we shall show that under certain practical conditions, the difference between the two fields, $\bm{E}$ in \eqref{eq:main1} and $\bm{E}_0$ in \eqref{eq:MaxwellFree} can be small; that is
\begin{equation}\label{eq:diff1}
\|\bm{E}-\bm{E}_0\|_{L^\infty(\Gamma\times(0, T])^3}\ll 1.
\end{equation}
If \eqref{eq:diff1} holds true, then clearly the inverse problem \eqref{invention} can be approximately replaced by the following one
\begin{equation}\label{invention2}
\bm{E}_0\big|_{\Gamma\times(0,T]}\longrightarrow \Lambda_{\bm{z}}.
\end{equation}
Here, we note that in \eqref{invention2}, there is no inhomogeneity presented. Next, we shall show that \eqref{eq:diff1} can indeed hold under certain conditions.

By \eqref{eq:charge}, \eqref{eq:point wave} and \eqref{eq:ecd}, we know that the electromagnetic field is time-harmonic with frequency $\omega_0$. Hence, we can assume that  (cf. \cite{Ned})
\begin{equation}\label{eq:th1}
\bm{E}(\bm x, t)=\Re(E(\bm x)e^{-\mathrm{i}\omega_0 t}),\quad \bm{H}(\bm x, t)=\Re(H(\bm x)e^{-\mathrm{i}\omega_0 t}),
\end{equation}
and
\begin{equation}\label{eq:th2}
\bm{E}_0(\bm x, t)=\Re(E_0(\bm x)e^{-\mathrm{i}\omega_0 t}),\quad \bm{H}_0(\bm x, t)=\Re(H_0(\bm x)e^{-\mathrm{i}\omega_0 t}).
\end{equation}
Furthermore, in order to show \eqref{eq:diff1}, we could only consider a fixed instant, say $t_0$. In the following, at $t_0$, we still use the notations $\epsilon(\bm x), \mu(\bm x), \sigma(\bm x)$ and $\Omega$ to denote the optical parameters of the inhomogeneity and its support respectively, and this should be clear from the context. By plugging \eqref{eq:th1} into \eqref{eq:main1}, one then has
\begin{equation}\label{eq:th3}
\begin{cases}
\nabla\times E-\mathrm{i}\omega_0\mu_0 H=0 \hspace*{1.2cm}  \mbox{in}\ \ \mathbb{R}^3,\smallskip\\
\nabla\times H+\mathrm{i}\omega_0\epsilon E=J+\sigma E\hspace*{0.4cm} \mbox{in}\ \ \mathbb{R}^3.
\end{cases}
\end{equation}
where $J(\bm x)$ signifies the electric current density at the instant $t_0$. It can be shown that $J(\bm x)$ takes the following form
\begin{equation}\label{eq:th4}
J(\bm x)=p_0 \delta(\bm x-\bm{z}_0).
\end{equation}

It is known that $E$ and $H$ in \eqref{eq:th3} satisfy the following so-called Silver-M\"uller radiation condition (cf. \cite{CK,Ned})
\begin{equation}\label{eq:silver}
\lim_{|\bm x|\rightarrow+\infty} \left(H\times\bm{x}-|\bm x| E\right)=0.
\end{equation}
Similarly, by plugging \eqref{eq:th2} into \eqref{eq:MaxwellFree}, one has
\begin{equation}\label{eq:th5}
\begin{cases}
\nabla\times E_0-\mathrm{i}\omega_0\mu_0 H_0=0 \hspace*{.5cm} \mbox{in}\ \ \mathbb{R}^3,\smallskip\\
\nabla\times H_0+\mathrm{i}\omega_0\epsilon_0 E_0=J \hspace*{.5cm} \mbox{in}\ \ \mathbb{R}^3,\smallskip\\
\displaystyle{\lim_{|\bm x|\rightarrow+\infty} \left(H_0\times\bm{x}-|\bm x| E_0\right)=0.}
\end{cases}
\end{equation}

We can show that
\begin{thm}\label{thm:main}
Let $(E, H)\in H_{loc}^1(\mathbb{R}^3)\times H_{loc}^1(\mathbb{R}^3)$ and $(E_0, H_0)\in H_{loc}^1(\mathbb{R}^3)\times H_{loc}^1(\mathbb{R}^3)$ be respectively the solutions to the Maxwell systems \eqref{eq:th3} and \eqref{eq:th5}. Then we have the following results.
\begin{enumerate}
\item[(i)] Assume that $\epsilon, \sigma\in C^{0,1}(\overline{\Omega})$. If $L:=\mathrm{dist}(\Gamma,\Omega)\gg 1$ and $L\omega_0\sim 1$, then
\begin{equation}\label{eq:ineq1}
\|E-E_0\|_{L^\infty(\Gamma)}\leq \alpha_0/L^2,
\end{equation}
where $\alpha_0$ is a positive constant depending only on $p_0, \bm{z}_0$ and $\mu_0$, $\epsilon_0$, $\|\epsilon\|_{C^{0,1}(\overline{\Omega})}$, $\|\sigma\|_{C^{0,1}(\overline{\Omega})}$.

\item[(ii)] Assume that $\epsilon, \sigma$ are constant in $\Omega$. If $\bm{z}_0$ is lying outside $\Omega$ and $\omega_0/c_0\ll 1$, then
\begin{equation}\label{eq:ineq2}
\|E-E_0\|_{L^\infty(\Gamma)}\leq \frac{\beta_0\omega_0}{c_0},
\end{equation}
where $c_0=1/\sqrt{\mu_0\epsilon_0}$ and $\beta_0$ is a positive constant depending only on $p_0, L$ and $\mu_0, \epsilon, \sigma$.
\end{enumerate}
\end{thm}

\begin{rem}\label{rem:1}
In order to give a completely rigorous justification of \eqref{eq:diff1}, one should make use of the Fourier transform to convert the time-domain problems \eqref{eq:main1} and \eqref{eq:MaxwellFree} into their frequency-domain counterparts. However, since the optical parameters $\epsilon$ and $\sigma$ may also vary according to time, one may meet difficulties in such a Fourier argument. Hence in Theorem~\ref{thm:main}, we simplify our study by using the time-harmonic assumption \eqref{eq:th1} as well as by considering a fixed instant $t_0$. Nevertheless, it is unobjectionably to claim that under the conditions specified in Theorem~\ref{thm:main}, \eqref{eq:diff1} should also hold true as long as the emitter is moving within a certain bounded domain.
\end{rem}

\begin{rem}\label{rem:2}
By Theorem~\ref{thm:main}, $(i)$, if the inhomogeneity is away from the measurement surface by a reasonable distance, then the scattering influence from the inhomogeneity can be neglected under low frequency emission. We note that $|E_0(\bm{x})|\sim 1/|\bm{x}-\bm{z}_0|$ for $|\bm{x}-\bm{z}_0|$ sufficiently large, and hence \eqref{eq:ineq1} indicates that the difference between $E$ and $E_0$ on $\Gamma$ is indeed nearly negligible. The inhomogeneity is mainly used to model the body of the human being who performs the instruction/input, and this means that one can keep a reasonable distance away from the computing device when performing the instruction/input. Next, in Theorem~\ref{thm:main}, $(ii)$, we note that the assumption on $\bm{z}_0$ is obviously a reasonable one. The assumption on $\epsilon$ and $\mu$ being constant in $\Omega$ can be relaxed as being with small variations. In such a case, one can choose a small frequency $\omega_0$, and in doing so, the scattering influence from the inhomogeneity can also be eliminated. All in all, Theorem~\ref{thm:main} indicates that one should make use of low-frequency emission, and the human being that performs the input/instruction should keep a reasonable distance away from the computing device in the design of the proposed gesture-computing device. This has been confirmed by our numerical experiments in what follows under practical and mild conditions on $\omega_0$ and $L$, not as restrictive as the theoretical assumptions in Theorem~\ref{thm:main}.
\end{rem}

\begin{proof}[Proof of Theorem~\ref{thm:main}]

We first recall that the fundamental solution to the following system
\begin{equation}\label{eq:fundamental}
\begin{cases}
\mathrm{i}\omega_0\epsilon_0 U+\nabla\times V=\delta I,\quad \mbox{$I$ is the $3\times 3$ identify matrix},\smallskip\\
-\mathrm{i}\omega_0\mu_0 V+\nabla\times U=0,
\end{cases}
\end{equation}
is given by (see Theorem~5.2.1 in \cite{Ned}),
\begin{equation}\label{eq:fundamental2}
\begin{cases}
U(\bm x)=&\mathrm{i}\omega_0\mu_0 G(\bm x)I+\frac{\mathrm{i}}{\omega_0\epsilon_0} \nabla^2 G(\bm x),\\
V(\bm x)=&\nabla\times(G(\bm x) I),
\end{cases}
\end{equation}
where
\[
G(\bm x):=\frac{1}{4\pi}\frac{e^{\mathrm{i}k_0|\bm x|}}{|\bm x|}\quad\mbox{with}\quad k_0:=\omega_0\sqrt{\epsilon_0\mu_0}.
\]
In \eqref{eq:fundamental2}, the second order derivation should be understood in the sense of distributions in $\mathcal{D}'(\mathbb{R}^3)$.
Next, by subtracting \eqref{eq:th5} from \eqref{eq:th3}, one has
\begin{equation}\label{eq:ff1}
\begin{cases}
\nabla\times (H-H_0)+\mathrm{i}\omega_0\epsilon_0 (E-E_0)=\sigma E+\mathrm{i}\omega_0(\epsilon_0-\epsilon) E,\\
\nabla\times (E-E_0)-\mathrm{i}\omega_0\mu_0 (H-H_0)=0.
\end{cases}
\end{equation}
Applying \eqref{eq:fundamental} and \eqref{eq:fundamental2} to \eqref{eq:ff1}, along with some straightforward calculations, one has the following integral representations
\begin{equation}\label{eq:ff2}
\begin{cases}
E= E_0+\displaystyle{\int_\Omega \left(\mathrm{i}\omega_0\mu_0 G(\bm x-\bm y)+\frac{\mathrm{i}}{\omega_0\epsilon_0}\nabla_{\bm{x}}^2 G(\bm x-\bm y) \right)\cdot \left(\sigma E(\bm{y})+\mathrm{i}\omega_0(\epsilon_0-\epsilon) E(\bm{y})\right)\, {\rm d}\bm{y} }\medskip \\
H= H_0+\displaystyle{\int_\Omega\nabla_{\bm{x}}(G(\bm x-\bm y) I)\cdot(\sigma E(\bm{y})+\mathrm{i}\omega_0(\epsilon_0-\epsilon) E(\bm y))\, {\rm d}\bm{y} }
\end{cases}
\end{equation}
By the well-posedness of the Maxwell system \eqref{eq:th3}--\eqref{eq:silver}, and a standard compactness argument, one has for $\omega_0<1$ that
\begin{equation}\label{eq:ff3}
\|E\|_{H(\mathrm{curl},\Omega)}+\|H\|_{H(\mathrm{curl},\Omega)}\leq C_1\quad\mbox{and}\quad \|\nabla\cdot E\|_{L^2(\Omega)}\leq C_2,
\end{equation}
where $C_1$ and $C_2$ are two positive constants depending only on $\Omega, p_0, \bm{z}_0$, and $\mu_0, \epsilon_0, \|\epsilon\|_{C^{0,1}(\overline{\Omega})}$, $\|\sigma\|_{C^{0,1}(\overline{\Omega})}$. One also has from the first equation in \eqref{eq:ff1} that in the sense of distribution,
\begin{equation}\label{eq:ff4}
\nabla\cdot(\sigma E)=\mathrm{i}\omega_0\nabla\cdot(\epsilon E-\epsilon_0 E_0)\quad \mbox{in}\ \ \Omega.
\end{equation}
Hence, by using \eqref{eq:ff4} and \eqref{eq:ff3}, we have
\begin{equation}\label{eq:ff5}
\begin{split}
&\,\left|\int_{\Omega} \frac{\mathrm{i}}{\omega_0\epsilon_0}\nabla_{\bm{x}}^2 G(\bm x-\bm y) \cdot ( \sigma E(\bm{y})+\mathrm{i}\omega_0(\epsilon_0-\epsilon) E(\bm{y}))\, {\rm d}\bm{y}\right|\\
=&\,\left| \int_{\Omega} \nabla_{\bm{x}} G(\bm{x}-\bm{y})\cdot \nabla (E-E_0)\, {\rm d}\bm{y} \right|\\
\leq &\, C_3/L^2
\end{split}
\end{equation}
for $\bm{x}\in\Gamma$ with $L:=\mathrm{dist}(\Gamma,\Omega)\gg 1$ and $L\omega_0\sim 1$, where $C_3$ is a positive constant depending only on $\Omega, p_0, \bm{z}_0$, and $\mu_0, \epsilon_0, \|\epsilon\|_{C^{0,1}(\overline{\Omega})}$, $\|\sigma\|_{C^{0,1}(\overline{\Omega})}$. Furthermore, by using $L:=\mathrm{dist}(\Gamma,\Omega)\gg 1$ and $L\omega_0\sim 1$, one can easily show
\begin{equation}\label{eq:ff6}
\left|\int_{\Omega}\mathrm{i}\omega_0\mu_0 G(\bm{x}-\bm{y})\cdot(\sigma E(\bm y)+\mathrm{i}\omega_0(\epsilon_0-\epsilon) E(\bm{y}))\, {\rm d}\bm{y} \right|\leq C_4/L^2,
\end{equation}
for $\bm{x}\in\Gamma$. Finally, by using the first equation in \eqref{eq:ff2} in estimating $\|E-E_0\|_{L^\infty(\Gamma)}$, together with the use of the estimates in \eqref{eq:ff5} and \eqref{eq:ff6}, one can readily arrives at the estimate in \eqref{eq:ineq1}.

For the second case in Theorem~\ref{thm:main}, it can be proved by following a completely similar argument as that for the proof of the first case. The main point is that if the emitter is lying outside $\Omega$, and $\epsilon, \sigma$ are constant in $\Omega$, one can easily have from the second equation in \eqref{eq:th3} that $\nabla\cdot E=0$ in $\Omega$. This fact can significantly simplify the corresponding argument in deriving \eqref{eq:ineq2}.

The proof is complete.

\end{proof}

In what follows, we let $c_0=1/\sqrt{\epsilon_0\mu_0}$ be the speed of the light in the background space. It is natural to assume that $|\bm v(t)|\ll c_0$; that is, the emitter is moving in a speed much slower than the light.
The retarded time $\tau$ is defined implicitly by the unique solution to
\begin{equation}\label{eq:ret}
\tau=t-\frac{|\bm R(\bm x,\tau)|}{c_0},\quad \bm x \in \Gamma,\,\, 0<\tau<t,
\end{equation}
where
$\bm R(\bm x,\tau):=\bm{x}-\bm{z}(\tau)$.
 Since $\Gamma \cap \Lambda_{\bm{z}}=\emptyset$, it holds that
 $|\bm R(\bm x,\tau)|>0, \forall x\in \Gamma, \tau \in \mathbb{R}_+$.
Then the solution $\bm{E}_0$ to \eqref{eq:MaxwellFree} is given by the well-known Li\'{e}nard-Wiechert potential (see, e.g. \cite[p. 438]{Griffiths})
\begin{equation}\label{eq:electric}
\begin{aligned}
\bm{E}_0(\bm x,t;\Lambda_{\bm{z}})
  =&\frac{q(\tau)}{4\pi\epsilon_0}\frac{|\bm{R}(\bm x,\tau)|}{(\bm{R}(\bm x,\tau) \cdot \bm u(\bm x,\tau))^3}\\
  &\cdot\left( (c_0^2-|\bm{v}(\tau)|^2)\bm u(\bm x,\tau)
  -\bm u(\bm x,\tau) \times \frac{\mathrm{d}^2\bm{z}(\tau)}{\mathrm{d}t^2}\times\bm{R}(\bm x,\tau)\right),
\end{aligned}
\end{equation}
where
\begin{equation*}
  \bm u(\bm x,\tau)=c_0 \hat{\bm{R}}(\bm x,\tau)-\bm{v}(\tau), \quad \hat{\bm{R}}(\bm x,\tau)=\frac{\bm{R}(\bm x,\tau)}{|\bm{R}(\bm x,\tau)|}.
\end{equation*}

\begin{thm}\label{thm2.2}
Let $\tau$ be the retarded time defined in \eqref{eq:ret} for $t\in (0,T]$ and $\bm{E}_0(\bm{x}, t;\Lambda_{\bm{z}})$ be defined in  \eqref{eq:electric}. Suppose that
\begin{equation}\label{eq:conn1}
\zeta(t):=\frac{(t-\tau)\omega_0}{2\pi}\ll 1,
\end{equation}
then we have
\begin{equation}\label{eq:appro}
\bm{E}_0(\bm{x},t;\Lambda_{z})=\frac{\sin \omega_0 t }{4\pi\epsilon_0|\bm R(\bm x,t)|^2} \hat{\bm{R}}(\bm x, \tau)+ \mathcal{O}(\zeta(t)) \quad\mbox{for}\ \ (\bm{x}, t)\in \Gamma\times (0, T].
\end{equation}
\end{thm}

\begin{proof}
From \eqref{eq:charge}, $\bm{E}_0(\bm{x}, t;\Lambda_{\bm{z}})$ defined in \eqref{eq:electric} can be written as
\begin{equation}\label{eq:1.2}
\begin{aligned}
\!\! \!\!\!\! &\bm{E}_0(\bm{x},t;\Lambda_{\bm{z}})=
\frac{\sin \omega_0 \tau }{4\pi\epsilon_0 |\bm{R}(\bm x,\tau)|^2(1-c_0^{-1}\hat{\bm{R}}(\bm x, \tau)\cdot\bm{ v}(\tau))^3} \\
\!\! \!\!\!\! &\cdot \left(\hat{\bm{R}}(\bm x, \tau)-\frac{\bm{v}(\tau)}{c_0}-\frac{\bm{v}^2(\tau)\hat{\bm{R}}(\bm x, \tau)}{c_0^2}
+ \frac{\bm{v}^3(\tau)}{c_0^3} -\frac{\bm u(\bm x,\tau) \times \displaystyle \frac{\mathrm{d}^2\bm{z}(\tau)}{\mathrm{d}t^2}\times \bm{R}(\bm x,\tau)}{c_0^3}\right),
\end{aligned}
\end{equation}
where $\tau$ is the retarded time defined in \eqref{eq:ret}.
By the assumption, we have
\begin{equation*}
  \displaystyle t-\tau=\zeta(t)\frac{2\pi}{\omega_0}.
\end{equation*}
 Next, by straightforward calculations, we have
\begin{equation}\label{eq:dd1}
\begin{split}
  \sin\omega_0\tau=&\sin \omega_0 \left(t-\zeta(t)\frac{2\pi}{\omega_0}\right)\\
                 =& \sin\omega_0 t \,\cos 2\pi\zeta(t)-\cos\omega_0 t \, \sin 2\pi\zeta(t)\\
                  =& \sin\omega_0 t +\mathcal{O}(\zeta(t)).\\
\end{split}
\end{equation}
and there exist $\eta_1, \eta_2, \eta_3$, such that $t-2\pi\zeta(t)/\omega_0<\eta_1, \eta_2, \eta_3 < t$ and
\begin{equation}\label{eq:dd2}
\begin{split}
 |\bm{R}(\bm x,\tau)|
              =&|\bm{x}-\bm{z}(\tau)|\\
               =&\left|\bm{x}-\bm{z}(t-\zeta(t)\frac{2\pi}{\omega_0})\right|\\
               =&\left|\bm{x}-\bm{z}(t)+\left(\frac{\mathrm{d}{z}_{1}(\eta_1)}
               {\mathrm{d} t},\frac{\mathrm{d}z_{2}(\eta_2)}
               {\mathrm{d} t},\frac{\mathrm{d}z_{3}(\eta_3)}
               {\mathrm{d} t}\right)\zeta(t)\frac{2\pi}{\omega_0}\right|\\
               =& |\bm{x}-\bm{z}(t)|\,|1+\mathcal{O}(\zeta(t))|,\\
\end{split}
\end{equation}
where $\bm{z}=(z_{1},z_{2},z_{3})$.
In addition,
\begin{equation}\label{eq:dd3}
\begin{split}
  \frac{\hat{\bm{R}}(\bm x,\tau)\cdot \bm{v}(\tau)}{c_0}
  =&\frac{|\bm{v}(\tau)|\cos\beta(\tau)}{c_0}\\
  =&\frac{2\pi\zeta(t)|\bm{v}(\tau)|\cos\beta(\tau)}{\omega_0|\bm{R}(\bm x,\tau)|}\\
  =&\mathcal{O}(\zeta(t)),
\end{split}
\end{equation}
where $\beta(\tau)\in [0, \pi]$ denotes the angle between  $\bm{x}-\bm{z}(\tau)$ and $\bm{v}(\tau)$.
Finally, by plugging \eqref{eq:dd1}, \eqref{eq:dd2} and \eqref{eq:dd3} into \eqref{eq:1.2}, along with straightforward asymptotic analysis, one can show
\begin{equation}\label{eq:dd4}
\bm{E}_0(\bm{x},t;\Lambda_{z})=\frac{\sin \omega_0 t }{4\pi\epsilon_0|\bm R(\bm x,t)|^2} \hat{\bm{R}}(\bm x, \tau)+ \mathcal{O}(\zeta(t)) \quad\mbox{for}\ \ (\bm{x}, t)\in \Gamma\times (0, T].
\end{equation}
The proof is complete.

\end{proof}

\section{Motion trajectory recovery}

We are now in a position to present two imaging schemes for qualitatively determining the motion trajectory $\Lambda_{\bm{z}}$ by knowledge of $|\bm{E}(\bm{x},t;\Lambda_{\bm z})|$ , namely,
\begin{equation}\label{invention}
  |\bm{E}(\bm{x},t;\Lambda_{\bm{z}})| \longrightarrow \Lambda_{\bm{z}},\quad (\bm{x}, t)\in \Gamma\times(0, T].
\end{equation}

We would like to point out that we are using the data without phase information, i. e., only the strength of the wave field is available. For inverse problems with phaseless data, we refer to the recent work \cite{KR}. 

\subsection{Imaging via the direct sampling method}

Define
\begin{equation}\label{eq:varphi}
  \phi(\bm x,t;{\bm y}):=\frac{|\sin \omega_0 t|}{4\pi\epsilon_0|\bm x-{\bm y}|^2},\quad  (\bm x,t,{\bm y})\in \Gamma\times(0,T] \times D,
\end{equation}
where $D\subset \mathbb{R}^3$ is the static compact sampling region, such that $\Lambda_{\bm{z}}\subset D$.
In the present study, for $(\hat{\bm{z}},t)\in D \times(0,T]$, we propose the indicator function as follows:
\begin{equation}\label{eq:indicator}
  I(\bm y, t):= \exp\left(-\frac{1}{\alpha|\Gamma|} \int_{\Gamma} \left(\phi(\bm x,t;{\bm y})- |\bm E(\bm x,t;\Lambda_{\bm{z}})|\right)^2 \mathrm{d}s(\bm x)\right),
\end{equation}
where $\alpha>0$ is a parameter and $|\Gamma|$ denotes the surface area of $\Gamma$.

\begin{thm}\label{thm3}
Let  $\bm E(\bm x,t;\Lambda_{\bm{z}})$ be the measurement data for $(\bm x,t)\in\Gamma\times(0,T]$, corresponding to a moving point charge described in \eqref{eq:main1} and let $I({\bm y},t)$ be defined in \eqref{eq:indicator}.
Define
\begin{equation}\label{eq:indicator0}
  I_0({\bm y}, t )=\exp\left(-\frac{1}{\alpha|\Gamma|} \int_{\Gamma}
 \left( \phi(\bm x,t;{\bm y})- |{\bm E_0}
 (\bm x,t;\Lambda_{\bm{z}})| \right)^2 \mathrm{d}s(\bm x)\right), \quad 0<t\leq T.
\end{equation}
Let the parameters be chosen such that \eqref{eq:ineq1} and \eqref{eq:ineq2} in Theorem \ref{thm:main} hold, then for each fixed $t_0\in (0, T]$ and any given $\epsilon>0$, there exists an $\omega_0>0$, such that  
\begin{equation*}
\begin{aligned}
 |I_0({\bm y}, t_0)-I({\bm y}, t_0 )|< C \varepsilon,\quad \forall {\bm y}\in D,
\end{aligned}
\end{equation*}
where $C>0$ is a constant depending on $t_0, \alpha$ and $|\Gamma|$. 
\end{thm}

\begin{proof}
Combining \eqref{eq:indicator} and \eqref{eq:indicator0}, one could obtain
\begin{equation}
\begin{aligned}
 &|I_0({\bm y}, t)-I({\bm y}, t )|\\
  &=\left|I_0({\bm y}, t ) \left(1 - \exp\left(-\frac{1}{\alpha|\Gamma|} \int_{\Gamma}
 \left( \phi- |{\bm E}| \right)^2-
 \left( \phi- |{{\bm E}_0}| \right)^2 \mathrm{d}s(\bm x)\right)  \right)\right|\\
  &=\left|I_0({\bm y}, t ) \left(1 - \exp\left(-\frac{1}{\alpha|\Gamma|} \int_{\Gamma}
 (2\phi- |{\bm E}|- |{\widetilde{\bm E}_0}|)
 (|{\widetilde{\bm E}_0}|- |{\bm E}|)\, \mathrm{d}s(\bm x)\right)  \right)\right|,\\
\end{aligned}
\end{equation}
using the Cauchy-Schwarz inequality,
\begin{equation*}
\begin{aligned}
&\int_{\Gamma}
 \left( 2\phi- |{\bm E}|- |{\bm E}_0|\right)
 \left(|{{\bm E}_0}|- |{\bm E}|\right) \mathrm{d}s(\bm x)\\
&=\int_{\Gamma}
 \left( 2(\phi-|{\bm E}_0|) - (|{\bm E}_0|-|{\bm E}|)\right)
 \left(|{{\bm E}_0}|- |{\bm E}|\right) \mathrm{d}s(\bm x)\\
&\leq
 \| 2(\phi-|{\bm E}_0|) - (|{\bm E}_0|-|{\bm E}|)\|_{L^2(\Gamma)}
  \| (|{{\bm E}_0}|- |{\bm E}| ) \|_{L^2(\Gamma)} \\
 & \leq
 (4\| \phi-|{\bm E}_0|\|_{L^2(\Gamma)} +2\| (|{\bm E}_0|-|{\bm E}|)\|_{L^2(\Gamma)})
  \| (|{{\bm E}_0}|- |{\bm E}| ) \|_{L^2(\Gamma)} \\
   & =
 4\| \phi-|{\bm E}_0|\|_{L^2(\Gamma)}\| (|{{\bm E}_0}|- |{\bm E}| ) \|_{L^2(\Gamma)} +2\| (|{\bm E}_0|-|{\bm E}|)\|_{L^2(\Gamma)}^2.
\end {aligned}
\end{equation*}
From \eqref{eq:dd4}, \eqref{eq:varphi} and Theorem \ref{thm:main},
for every $\varepsilon>0$, there exists an $\omega_0>0$, such that
\begin{equation*}
  \int_{\Gamma}
 \left( 2\phi- |{\bm E}|- |{\bm E}_0|\right)
 \left(|{{\bm E}_0}|- |{\bm E}|\right) \mathrm{d}s(\bm x)<\varepsilon.
\end{equation*}
Through \eqref{eq:appro} and \eqref{eq:varphi}, we can find that $I_0({\bm y},t_0)$
is bounded, i.e., $ |I_0|<\infty$. Let $C_1=\max_{{\bm y}\in D} I_0({\bm y},t_0)$. Then by Theorem \ref{thm:main}, Theorem \ref{thm2.2} and Taylor's expansion, we have
\begin{equation*}
 |I_0({\bm y}, t_0)-I({\bm y}, t_0 )|\leq |C_1(1 - \exp(-C_2\varepsilon))|< C_1 C_2 \varepsilon,
\end{equation*}
where $C_2=1/(\alpha |\Gamma|)$. 
\end{proof}

\begin{rem}\label{remark:indicator}
Theorem \ref{thm3} illustrates that $I({\bm y}, t)$ is quite close to $I_0({\bm y},t)$. Thus the indicator $I$ should inherit $I_0$'s maximum indicating behavior. For any fixed $t_0\in (0, T]$, it can be easily deduced that $I_0({\bm y},t_0)$ attains its maximum when ${\bm y}={\bm z}(t_0)$. Therefore, $I({\bm y},t_0)$ attains its maximum when ${\bm y}\approx{\bm z}(t_0)$.
\end{rem}

Based on Theorem \ref{thm3} and Remark \ref{remark:indicator}, we proposed  the first trajectory reconstruction  scheme for the inverse problem \eqref{invention}, see {\bf Algorithm 1}.
\begin{table}[htp]
\centering
\begin{tabular}{cp{.8\textwidth}}
\toprule
\multicolumn{2}{l}{{\bf Algorithm 1:}\quad Reconstruction of the trajectory via direct sampling}\\
 \midrule
 {\bf Step 1} & Properly choose a low frequency $\omega_0$ and a point charge of the form \eqref{eq:charge}.  \\
{\bf Step 2} & Set the emitter in motion following a specific path $\Lambda_{\bm{z}}$, depending on the desired input/instruction. The sensors collect the electric field data $|\bm{E}(\bm{x}_m, t_n;\Lambda_{\bm z})|$   at the measurement points $\{\bm{x}_m\}\in \Gamma$ and a sequence of discrete time points $\{t_n\}\in (0, T]$. \\
{\bf Step 3} & Select a  sampling mesh $\mathcal{T}_h$ in a region $D$ such that $\Lambda_z\subset D$. For each time point $t_n$, evaluate the imaging functional $I(\bm{y},t_n)$ defined in \eqref{eq:indicator} for each $\bm{z}\in \mathcal{T}_h$.\\
 {\bf Step 4} & Locate the global maximum point ${\bm z}_n$ of $I(\bm{y},t_n)$ for $\bm{y}\in\mathcal{T}_h$, which is an approximation to $\bm{z} (t_n)$. \\
{\bf Step 5} & The ordered chain $\{\bm z_n\}_{n=1}^{N_t}$ forms a discrete version of the reconstruction of $\Lambda_{\bm{z}}$. \\
\bottomrule
\end{tabular}
\end{table}
It is remarked that the proposed sampling scheme could be able to recover the moving trajectory if the sampling grid is sufficiently fine. However, the computational cost would be relatively high when the indicator functional is evaluated over a very fine mesh in {\bf Step 3}. To speed up the reconstruction process, in the next subsection, we shall develop an alternative meshless approach, namely, the modified particle filter method.

\subsection{Imaging via the modified particle filter method}

In this subsection, we are going to reformulate the problem of trajectory reconstruction as a stochastic inverse problem based on a probability space. The method provided here could be considered as a modified version of the classical particle filter method. For the ease of the readers, we recall some essential rudiments of the particle filter method in Appendix \ref{App}. 

The underlying assumption of the stochastic inverse problem is that the time-discrete reconstruction is a sequence of states of random variables, which sample some particular probability distribution.
Assume that the measured data were collected at $N_t$ discrete times $t_n = n\Delta t, n = 1, 2, \cdots, N_t$, with time step $\Delta t = T/N_t$. Now let us consider a time-discrete Markov chain of states $\{\bm{\xi}_n\}$ corresponding to time $t_n$ for $n=1,\cdots, N_t$. The transition probability of the Markov chain is denoted by $p(\bm{\xi}_n\mid\bm{\xi}_{n-1})$, i.e., the probability to draw $\bm{\xi}_n$ at time $t_n$ when $\bm{\xi}_{n-1}$ was drawn in its previous step $t_{n-1}$.

To specifically quantify the trajectory of the point emitter by the strategy of Markov chains,  we consider the following Gaussian random walk model $\{\bm{\xi}_n\}$ satisfying
\[
\bm{\xi}_n = \bm{\xi}_{n-1}+\bm{G}_n, \quad n=1, 2,\cdots.
\]
Here $\bm{G}_n$ denotes a Gaussian distribution with zero mean and covariance matrix
\[
\Upsilon: ={\rm diag}(\gamma^2, \gamma^2,\gamma^2)
\]
where $\gamma>0$ controls the step size in the spatial evolution.
Correspondingly, the transition density distribution is
\begin{equation}\label{eq:pro1}
 p(\bm{\xi}_{n}| \bm{\xi}_{n-1})=\frac{1}{(2\pi\gamma^2)^{3/2}} \exp\left(-\frac{1}{2}(\bm{\xi}_{n}- \bm{\xi}_{n-1})^{\top} \Upsilon^{-1}(\bm{\xi}_{n}- \bm{\xi}_{n-1})\right).
\end{equation}

Considering the observation model
\begin{equation*}
  |\bm E(\bm x,t_n;\Lambda_{\bm{z}})|=\phi(\bm x, t; \bm \xi_n)+ W_n,
\end{equation*}
in light of  \eqref{eq:indicator},  we define the  density function of $W_n$ by
\begin{equation}\label{eq:density1}
  p(|\bm E(\bm x,t_n;\Lambda_{\bm{z}})|\mid \bm \xi_n)= \frac{I(\bm \xi_n,t_n)}{\int_D I(\bm \xi_n,t_n)\,\mathrm{d}\bm \xi_n}.
\end{equation}
where $I$ is defined in \eqref{eq:indicator}.

By Appendix \ref{App}, the posterior probability density function is approximate to
 \begin{equation*}
 p^{N_s}(\bm \xi_n\mid |\bm E(\bm x,t_{1:n};\Lambda_{\bm{z}})|)
 =\frac{1}{{\rm d} \bm \xi_n}\sum\limits_{i=1}^{N_s}p(|\bm E(\bm x,t_n;\Lambda_{\bm{z}})|\mid \bm \xi_n) \delta_{\bm \xi_n^{(i)}}({\rm d} \bm \xi_n).
 \end{equation*}
where $\bm{\xi}_n^{(i)}\in D, 1\leq i\leq N_s,$ are $N_s$ random sample points at the instance $t_n$.

Based on the above discussion, we next present the trajectory reconstruction scheme via the modified particle filter method in {\bf Algorithm 2}.
\begin{table}[htp]
\centering
\begin{tabular}{cp{.8\textwidth}}
\toprule
\multicolumn{2}{l}{{\bf Algorithm 2:}\quad Reconstruction of the trajectory via  particle filter.}\\
\midrule
 {\bf Step 1} & Properly choose a low frequency $\omega_0$ and a point charge of the form \eqref{eq:charge}.  \\
{\bf Step 2} & Set the emitter in motion following a specific path $\Lambda_{\bm{z}}$, depending on the desired input/instruction. The sensors collect the electric field data $|\bm{E}(\bm{x}_m, t_n;\Lambda_{\bm{z}})|$   at the measurement points $\{\bm{x}_m\}\in \Gamma$ and a sequence of discrete time points $\{t_n\}\in (0, T]$. \\
 {\bf Step 3} &Initialization: for $n=0$, draw  an initial random sample $ \{\bm{\xi} _{n}^{(i)}\}_{i=1}^{N_s} $ from a  uniform distribution in a region $D$ such that $\Lambda_z\subset D$.   \\
{\bf Step 4} & Importance sampling: for $n\geq 1$,
let $\tilde{\bm \xi}_n^{(i)}=\bm \xi_{n-1}^{(i)}+b_n$, where $b_n\sim \mathcal{N}(0, \Upsilon)$
and calculate the relative likelihood\\
& \quad\quad\quad\quad $
w_{m,n}^{(i)}=\dfrac{p(|\bm E(\bm x_m,t_n;\Lambda_{\bm{z}})|\mid \tilde{\bm \xi}_n^{(i)})}{\sum\limits_{i=1}^{N_s} p(|\bm E(\bm x_m,t_n;\Lambda_{\bm{z}})|\mid \tilde{\bm \xi}_n^{(i)})}.
$\\
{\bf Step 5} & Resample: let $\mathcal{U}$ be the uniformly distribution density function,  draw a random number $q\sim \mathcal{U}([0,1])$. For every $i$,  set $\bm{\xi}_{n}^{(i)}=\tilde{\bm{\xi}}_{n}^{(\ell)}$ when
 $\displaystyle\sum_{i=1}^{\ell-1}w_{m,n}^{(i)}<q\leq\sum_{i=1}^{\ell}w_{m,n}^{(i)}$.
  Set $\bm z_n=\dfrac{1}{N_s}\sum_{i=1}^{N_s} \bm \xi_n^{(i)}$. \\
{\bf Step 6} &If $(n+1)\Delta t \geq N_t$, then the reconstruction is finished. Otherwise, set $n=n+1$ and repeat from \bf {Step 4}. \\
{\bf Step 7} & The ordered chain $\{\bm z_n\}_{n=1}^{N_t}$ forms a discrete version of the reconstruction of $\Lambda_{\bm{z}}$. \\
\bottomrule
\end{tabular}
\end{table}

\begin{rem}
The critical feature of the modified particle filter method is that we use multi-measurements to identify the trajectory at discrete instants. Using the discrete observation points, the discrete version of  the indicator function in \eqref{eq:indicator} could be written as
\begin{equation*}
\begin{aligned}
  {\bm I}({\bm y}, t):
  &=\exp\left(-\frac{1}{\alpha|\Gamma|} \sum_{m=1}^{N_m} \left(\phi(\bm x_m,t;{\bm y})- |\bm E(\bm x_m,t;\Lambda_{\bm{z}})|\right)^2 \frac{|\Gamma|}{N_m}  \right)\\
   &= \exp\left(-\frac{1}{\alpha N_m} \sum_{m=1}^{N_m} \left(\phi(\bm x_m,t;{\bm y})- |\bm E(\bm x_m,t;\Lambda_{\bm{z}})|\right)^2 \right).
  \end{aligned}
\end{equation*}
\end{rem}

\section{Numerical examples}\label{sect:numerical}

In this section, we will present several numerical examples to demonstrate the feasibility and effectiveness of the proposed method.

  All the following numerical experiments are carried out using  MATLAB R2016a   on a Lenovo workstation with $2.3$GHz Intel Xeon E5-2670 v3 processor and 128GB RAM. Synthetic  electromagnetic  field  data are generated by solving direct problem \eqref{eq:main1}  by using the quadratic finite elements on a truncated  spherical  domain enclosed by absorbing boundary  condition. The mesh of the forward solver is successively refined till the relative error of the successive measured electromagnetic wave data is below $0.1\%$. To test the stability of the proposed reconstruction algorithm, Gaussian noise was point-wisely added to the synthetic data, that is,
\begin{equation*}
  |\bm E_{\varepsilon}(\bm x,t;\Lambda_{\bm{z}})|:=|\bm E(\bm x,t;\Lambda_{\bm{z}})|(1+\varepsilon r ), \quad r\sim \mathcal{N}(0, 1),
 \end{equation*}
where $\varepsilon>0$ is the noise level. Here, $10\%$ noise was added to the  synthetic data, namely, $\varepsilon=0.1$.

\begin{figure}
\centering
\includegraphics[width=0.5\textwidth, clip, trim=50 200 50 200]{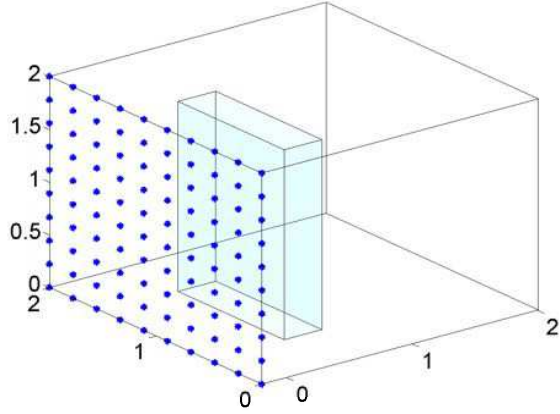}
\caption{\label{fig:geometry} The measurement points and human being's motion domain.}
\end{figure}

The physical quantities are used with SI units. Let $\omega_0=10^5$ Hz, $\mu_0=4\pi \times 10^{-7}$ H/m, $\epsilon_0=8. 85 \times 10^{-12}$ F/m and  the relative  permittivity of $\Omega$ be  $\epsilon_r=81.5$. For simplicity, if not otherwise specified, let $D=[0,\, 2\, {\rm m}]^3$ be divided into a uniformly distributed sampling mesh with dimension $N_h=50\times50\times50$. Assume that the human body moves inside a cubic domain $\Omega=[0.4 \mathrm{m}, 0.7 \mathrm{m}]\times[0.5\, \mathrm{m}, 1.5\, \mathrm{m}]\times[0\, \mathrm{m}, 1.8\, \mathrm{m}]$.
There are 400 uniformly distributed sensors $x_m\in \Gamma, \, m=1,2,\cdots, 400,$ and $N_t$ equidistant time steps $t_n=nT/N_t\in (0,T],\, n=1,2,\cdots,N_t$. The observation surface is set to be $\Gamma=-$0.2\, m$\times$ [0\, m,\, 2\, m]$\times$ [0\, m,\, 2\, m]. See Figure \ref{fig:geometry} for the geometry setting. At last, let $N_m=400$ and $N_t=T/\Delta t$, where the time step $\Delta t$ is set to be 0.1 s.  Assume that the average velocity of the point emitter is 0.5 m/s, hence the step size parameter $\gamma$ is chosen as $0.5\ {\rm m/s} \times\Delta t=0.05$ m.   Heuristically, we would like to remark that in our proposed imaging methods, the parameter $\alpha$ is such chosen that $\alpha\sim c_0^2$, then they can produce fine reconstructions. Thus $\alpha=10^{16}$ is used. Finally, as a convention in the following figures, the blue points denote the measurement points and the blue cubic domain denotes the person's motion domain. 

To show the accuracy of the proposed methods, we also define the discrete relative $L^2$ error between the exact trajectory and the reconstruction by
\begin{equation*}
\frac{\left(\sum_{n=1}^{N_t}|\bm z(t_n)-{\bm z}_n|^2\right)^{1/2}}{\left(\sum_{n=1}^{N_t}|\bm z(t_n)|^2\right)^{1/2}}.
\end{equation*}


\begin{example}\label{exp1}
In the first example, we compare the direct sampling method and the particle filter technique in reconstructing a moving trajectory.
Consider a simplified scenario that a person is wearing an emitter on one of his/her finger and moving the finger to write an  Arabic number ``3", which is modelled by the trajectory:
\begin{equation}\label{eq:number3}
  \bm{z}(t)=\left(1, \, 1.5-\left|\sin\frac{\pi}{5}t\right|,\, 1.8-0.16 t\right), \quad t\in(0,\,10\, {\rm s}].
\end{equation}
\end{example}

\begin{figure}
\centering
\hfill\subfigure[]{\includegraphics[width=0.45\textwidth]{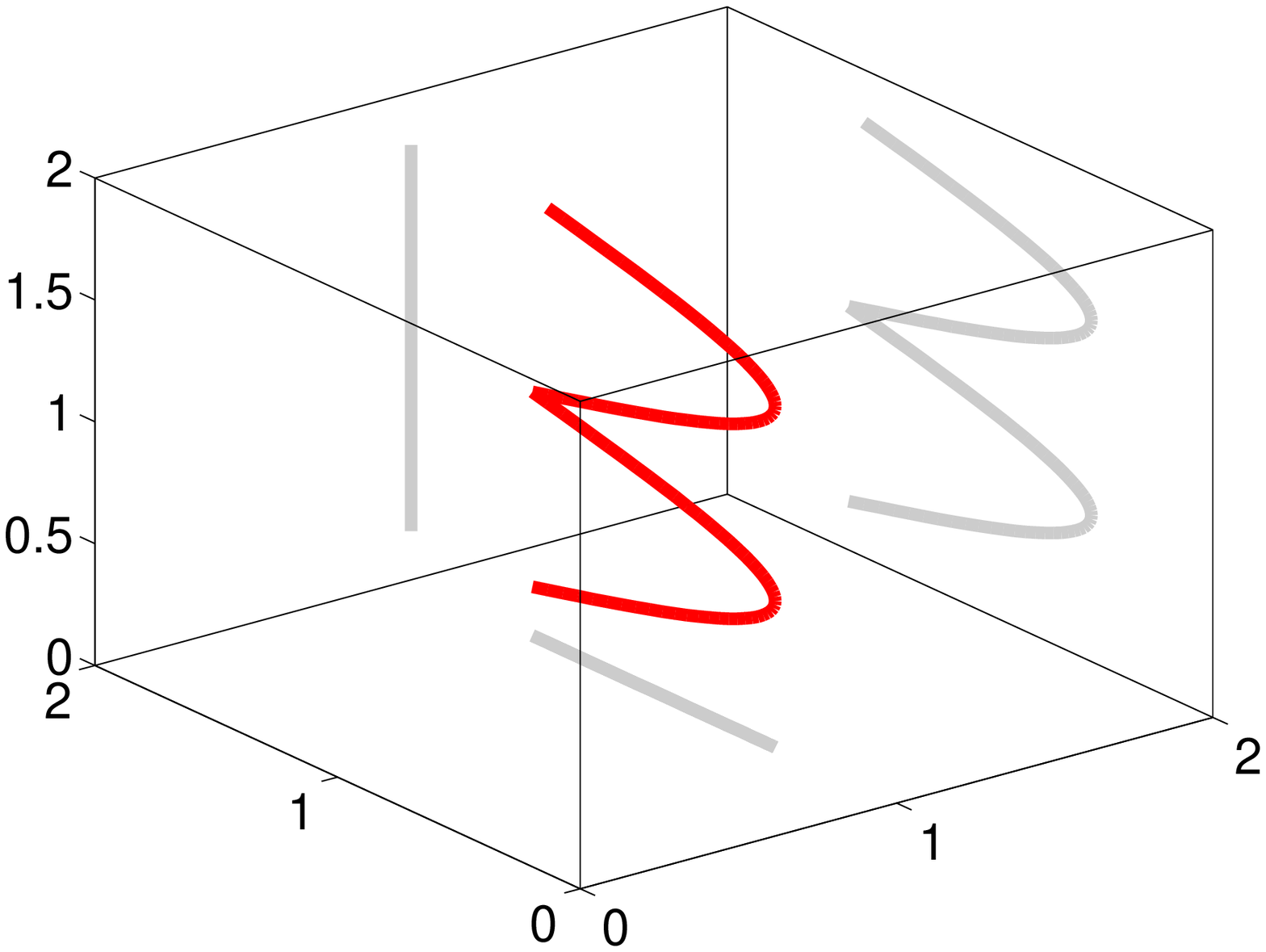}}\hfill 
\hfill\subfigure[]{\includegraphics[width=0.45\textwidth]{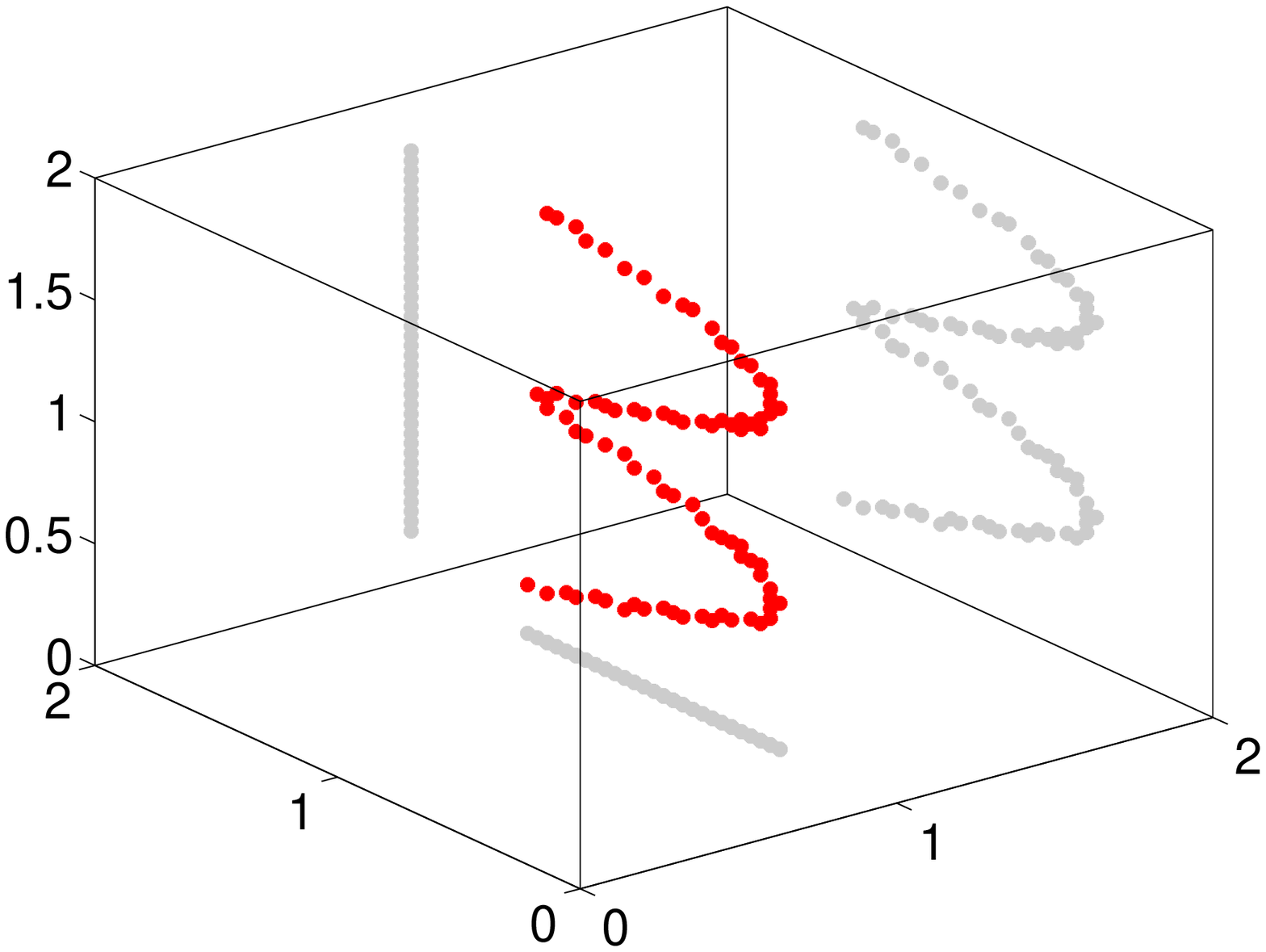}} \hfill \\
\hfill\subfigure[]{\includegraphics[width=0.45\textwidth]{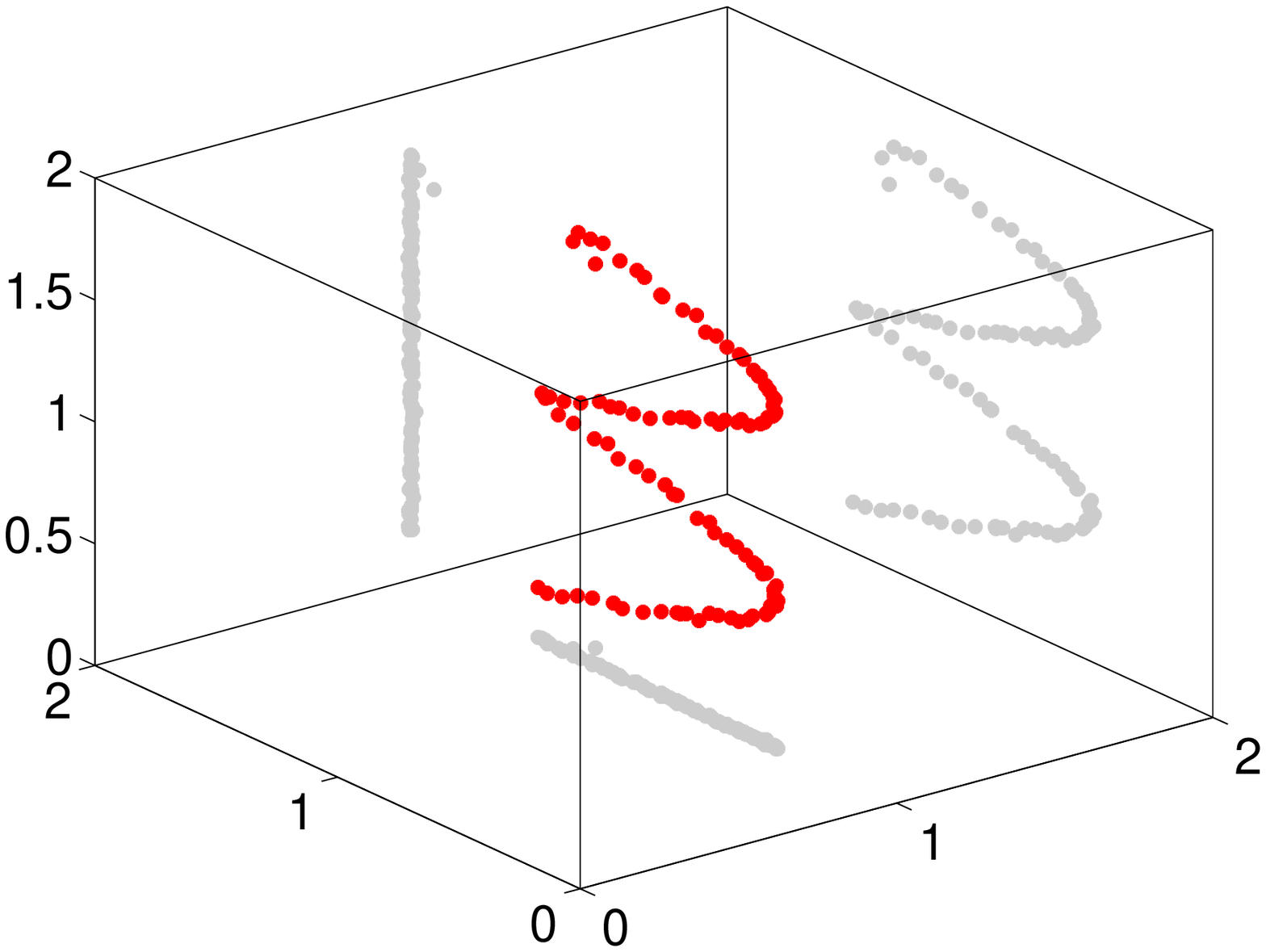}}\hfill 
\hfill\subfigure[]{\includegraphics[width=0.45\textwidth]{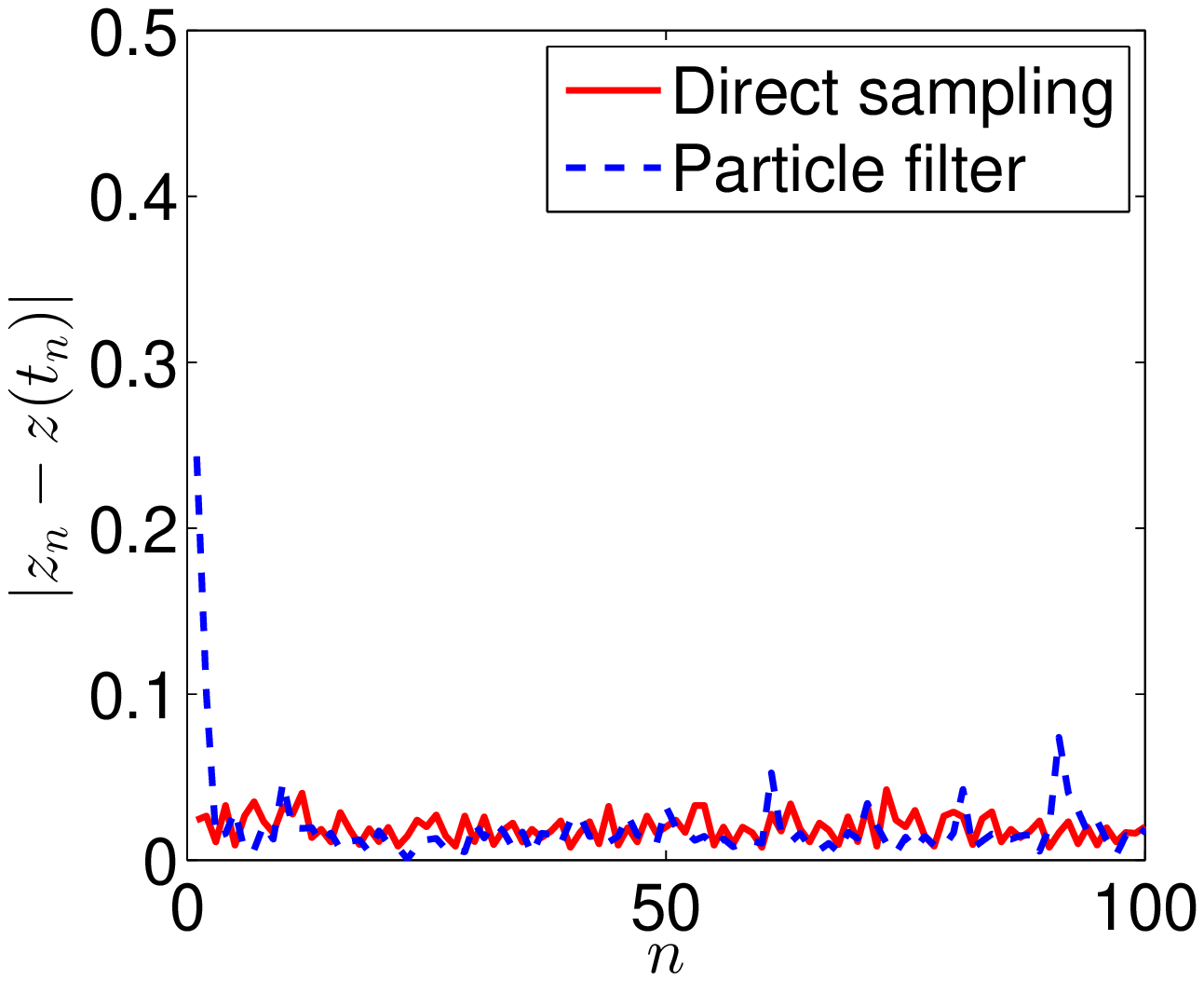}}\hfill
\caption{\label{fig:ex1} Reconstruction of the trajectory ``3".
(a) exact moving trajectory,
(b) reconstruction by the direct sampling method,
(c) reconstruction by the particle filter method with $N_s=100$,
(d) point-wise error between the exact and reconstructed trajectory.}
\end{figure}

The indicating scatter plots of both direct sampling method and particle filter method are shown in Figure  \ref{fig:ex1}.
In the following figures concerning the problem geometry, some 2D projections (shadows with gray color) are also added  in order to facilitate the 3D visualizations.
 Both methods are able to identify the moving trajectory even if the measured data has $10\%$ noise. Through comparing the error between the exact trajectory and the reconstructed trajectory in discrete instants, one can find that the recovery results of both methods have similar error fluctuations, see Figure \ref{fig:ex1}(d).
In order to exhibit the accuracy and computational cost of the reconstructions,  the relative error in discrete $L^2$  norm and the CPU time are listed in Table  \ref{tab1}.
\begin{table}[H]
\center
 \caption{The relative $L^2$ error and the CPU time for reconstructing the trajectory ``3'' ($N_h$ and $N_s$ denote, respectively, the number of sampling points and the number of particle samples).}\label{tab1}
 \begin{tabular}{lccccrc}
  \toprule
 {}&\multicolumn{3}{c}{\underline{\qquad Direct sampling method\qquad}}
 {}&\multicolumn{3}{c}{\underline{\qquad Particle filter technique\qquad}}  \\
 \quad   & $N_h=25^3$  &$N_h=50^3$   & $N_h=100^3$   & $N_s=50$   & $N_s=100$  & $N_s=500$   \\
 \midrule
 error  & 2.09\% & 0.79\%  & 0.68\%   &6.76\%  &2.62\%  &1.59\% \\
 CPU time &16 s     & 115 s    & 897 s  & 0.23 s  & 0.48 s  & 10.5 s  \\
  \bottomrule
 \end{tabular}
\end{table}
\begin{example}\label{exp2}
Reconstruction of a conical spiral. In the previous example, we only study moving trajectory in a plane.  This example is intended to demonstrate the performance of both methods for reconstructing a complicated 3D trajectory. Here the moving trajectory is like an upward conical spiral,
\begin{equation}\label{eq:conical}
  \bm{z}(t)=\left(\,0.05t\cos2t +1,\, 0.05t\sin2t+1,\, 0.1t+0.5\,\right), \quad t\in(0,10\,\rm{s}].
\end{equation}
\end{example}
From equation \eqref{eq:conical}, it is clearly that  the velocity $|\bm v(t)|$ is monotone increasing. As shown in Figure \ref{fig:ex2}(b)--(d),  both methods could produce satisfactory reconstructions.
In addition, Figure \ref{fig:ex2}(e) shows that the particle filter method works better with a higher sampling density, namely, the number of particles. However, the computational  cost increases as the number of particles increases, so a proper number of particles should be considered.

\begin{figure}
\centering
\subfigure[]{\includegraphics[width=0.45\textwidth]{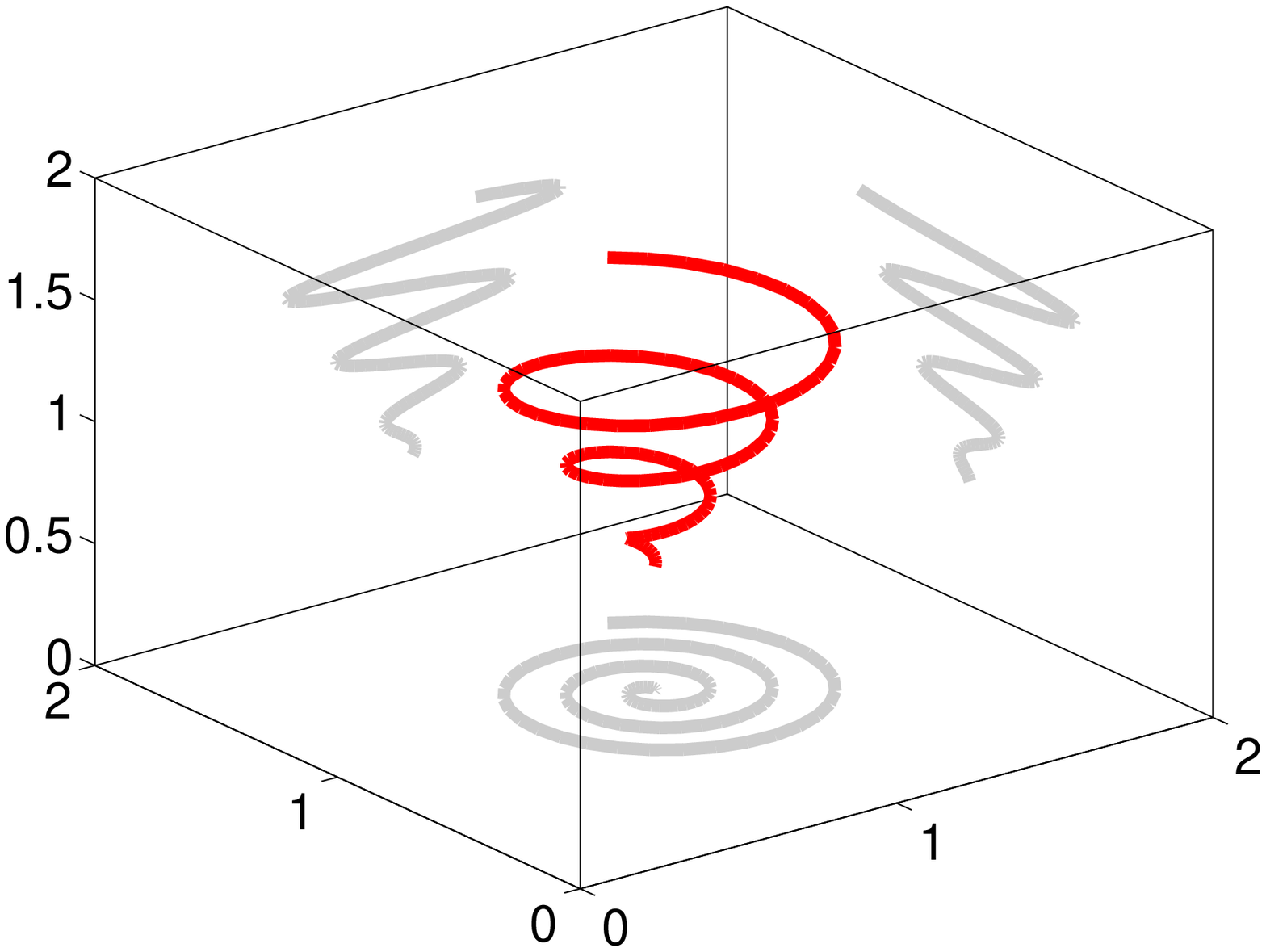}}\\
\hfill\subfigure[]{\includegraphics[width=0.45\textwidth]{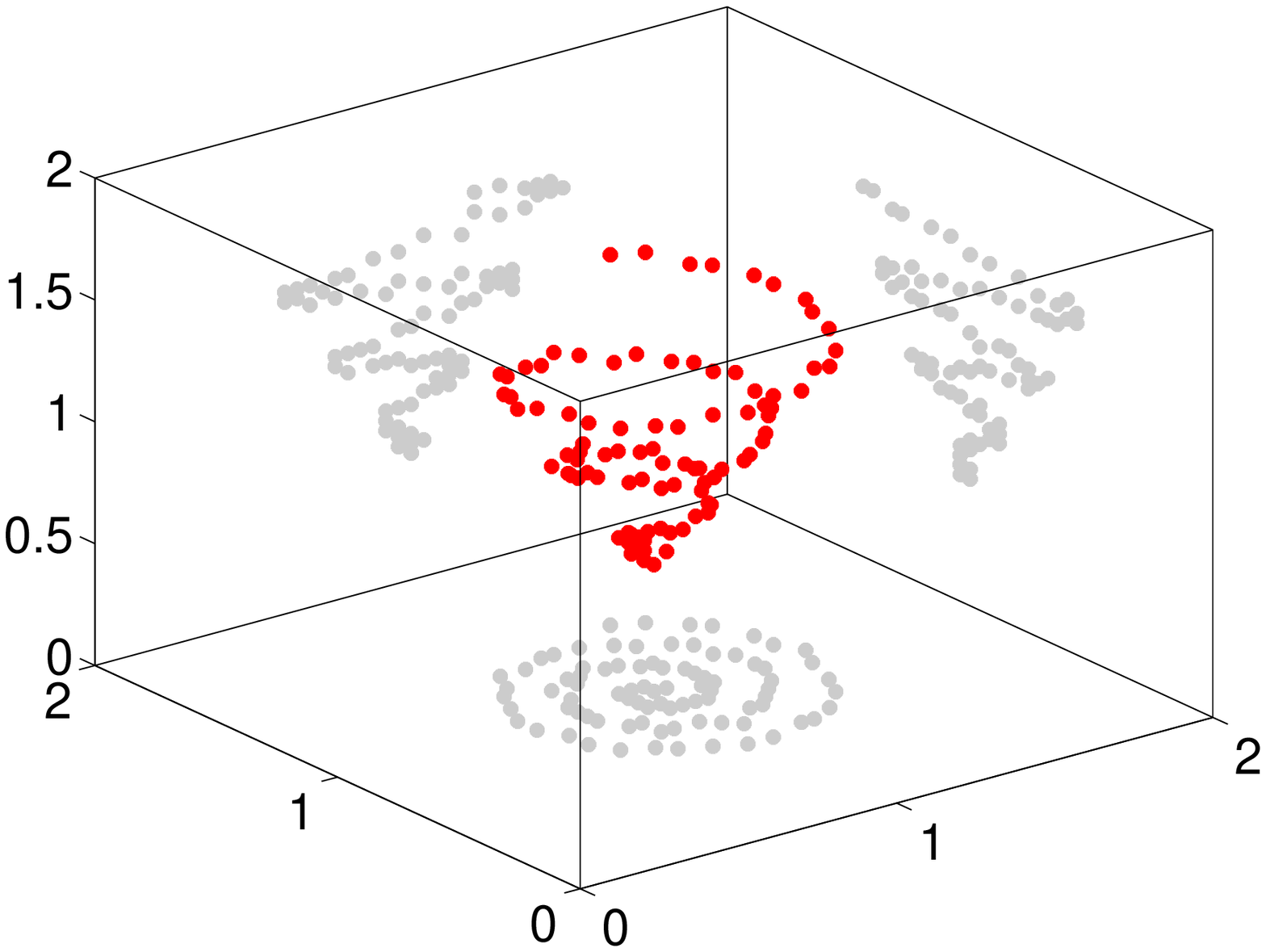}}\hfill
\hfill\subfigure[]{\includegraphics[width=0.45\textwidth]{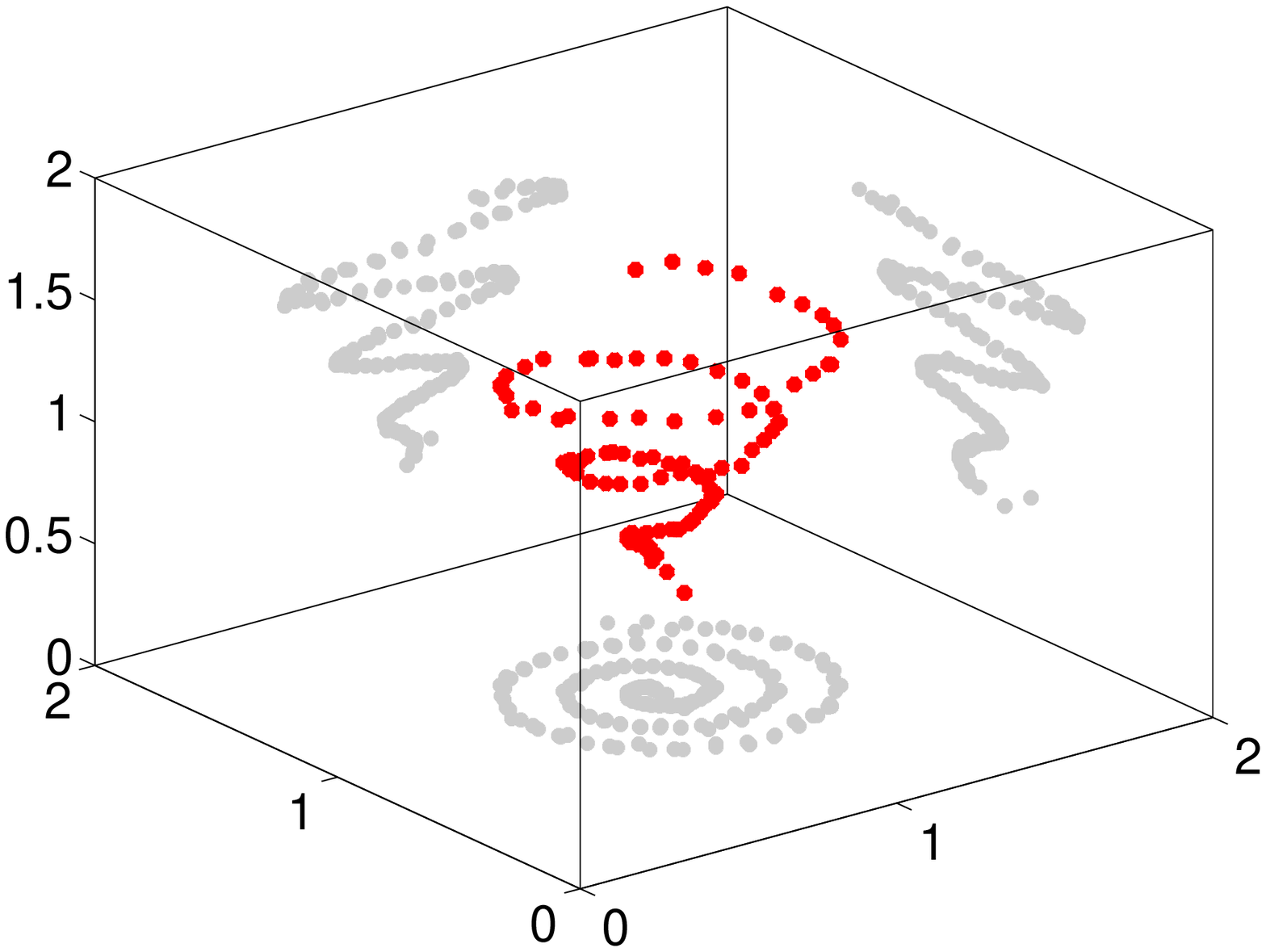}} \hfill\\
\hfill\subfigure[]{\includegraphics[width=0.48\textwidth]{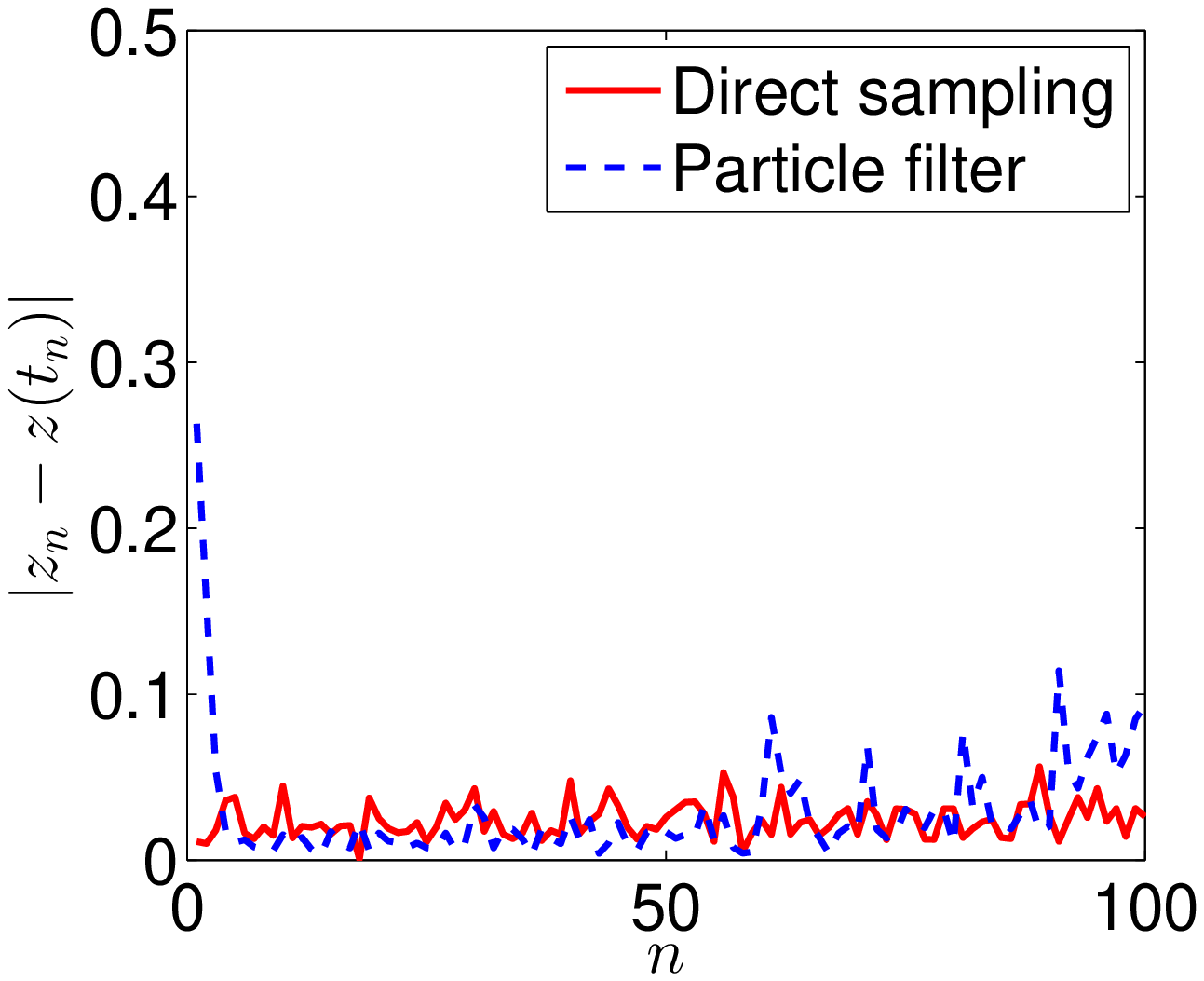}}\hfill
\hfill\subfigure[]{\includegraphics[width=0.4\textwidth]{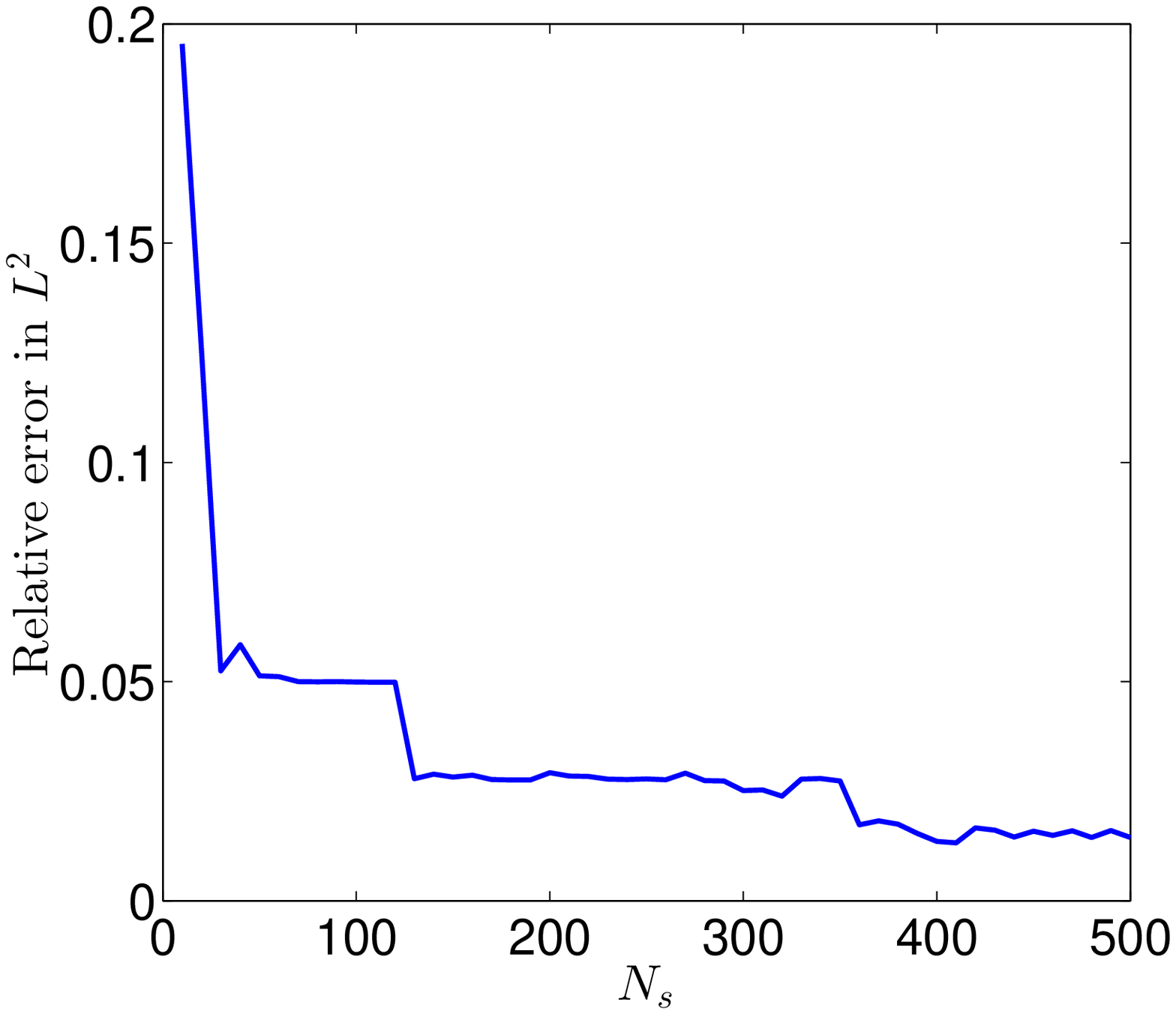}}\hfill
\caption{\label{fig:ex2} Reconstruction of a conical spiral shaped trajectory.
(a) exact moving trajectory,
(b) reconstruction result by the direct sampling method,
(c) reconstruction by the particle filter method with $N_s=200$,
(d) point-wise error between the exact and reconstructed trajectory,
(e) relationship between the number of particles and relative $L^2$ error.}
\end{figure}


\begin{example}\label{exp3}
Reconstruction of the trajectory of a text. The last example is devoted to identification of a motion trajectory consisting of handwriting letters ``OK'', where the moving path is as follows
\begin{equation*}
\bm{z}(t)=\left\{
\begin{aligned}
&  \left(\frac{-3\sqrt{2}}{20}\sin t+\frac{2\sqrt{2}}{5},\, \frac{3\sqrt{2}}{20}\sin t+2-\frac{2\sqrt{2}}{5},\, \frac{3}{5}\cos t+1\right),\quad t\in ( 0,\,6\,\rm{s}],\\
&  \left(1 , \, 1, \, -\frac{t}{2}+\frac{24}{5} \right),\quad t\in (6\,\rm{s},\,9\,\rm{s}],\\
&  \left(\frac{4\sqrt{2}}{25}t+1-\frac{21\sqrt{2}}{125} ,\, -\frac{4\sqrt{2}}{25}t+1+\frac{21\sqrt{2}}{125} ,\,-\frac{2}{5}t+\frac{26}{5} \right),\quad t\in( 9\,\rm{s},\,12\,\rm{s}].\\
\end{aligned}
\right.
\end{equation*}
\end{example}

\begin{figure}
\hfill\subfigure[]{\includegraphics[width=0.45\textwidth]{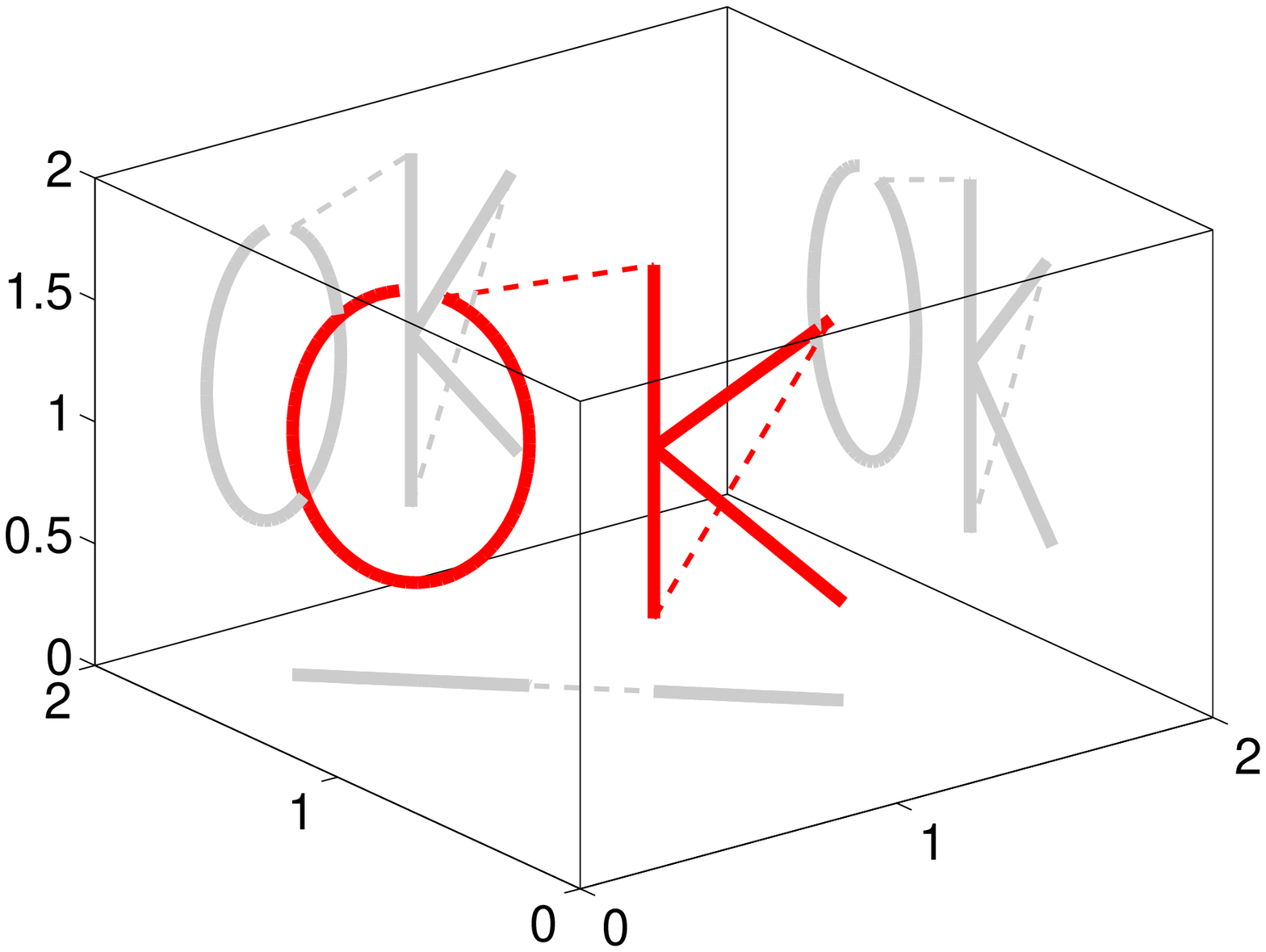}}\hfill
\hfill\subfigure[]{\includegraphics[width=0.45\textwidth]{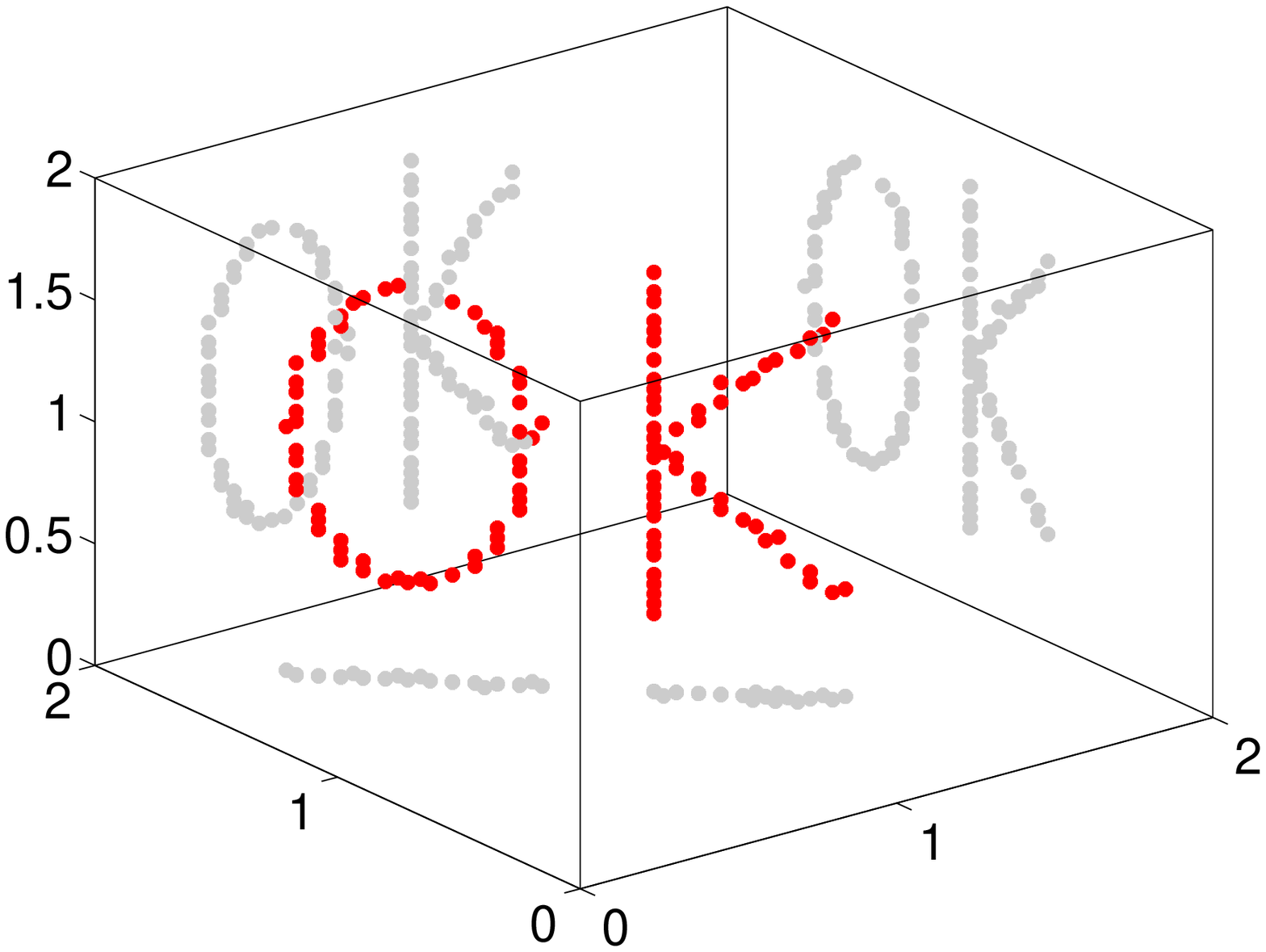}} \hfill\\
\hfill\subfigure[]{\includegraphics[width=0.45\textwidth]{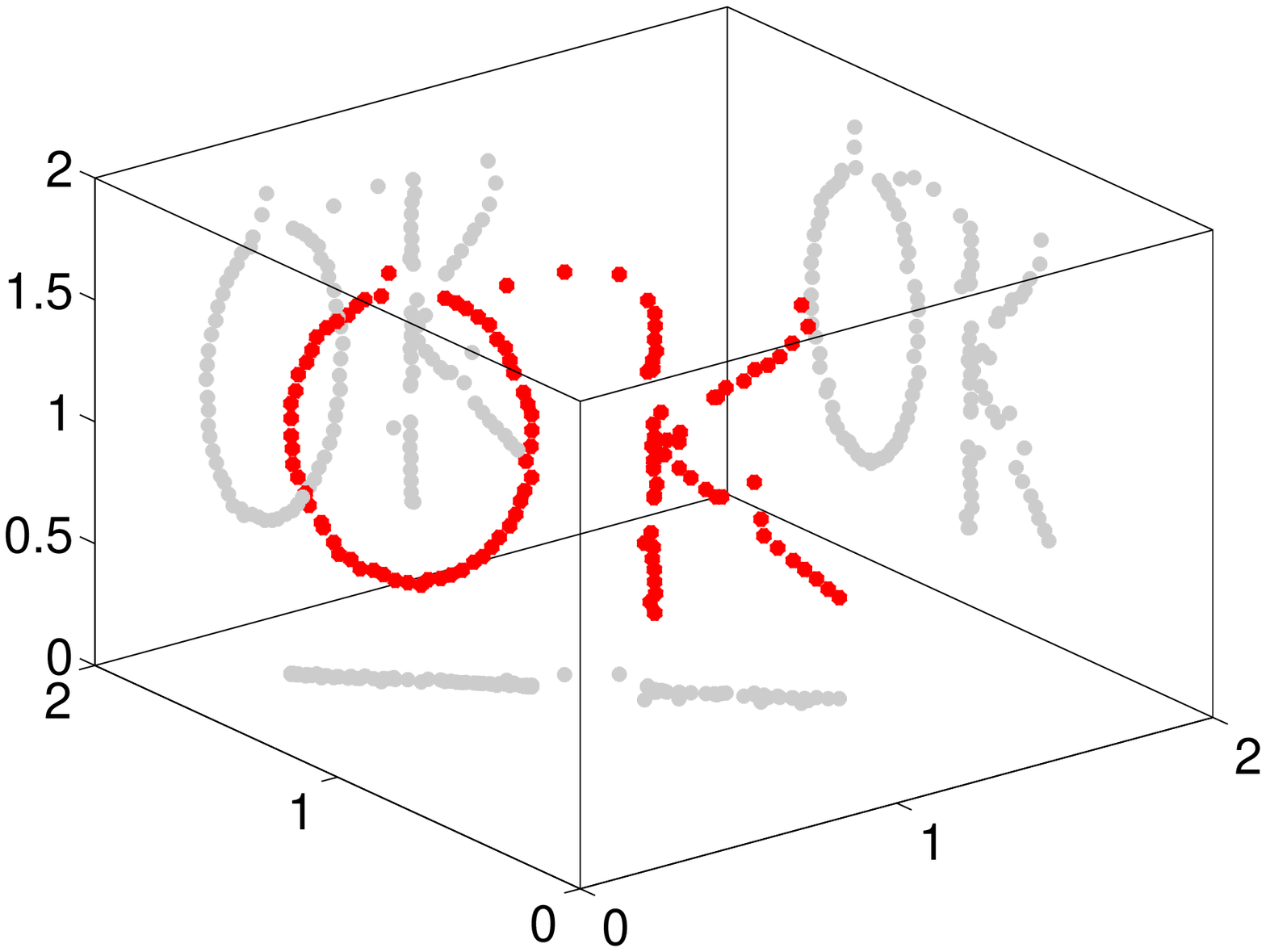}}\hfill
\hfill\subfigure[]{\includegraphics[width=0.45\textwidth]{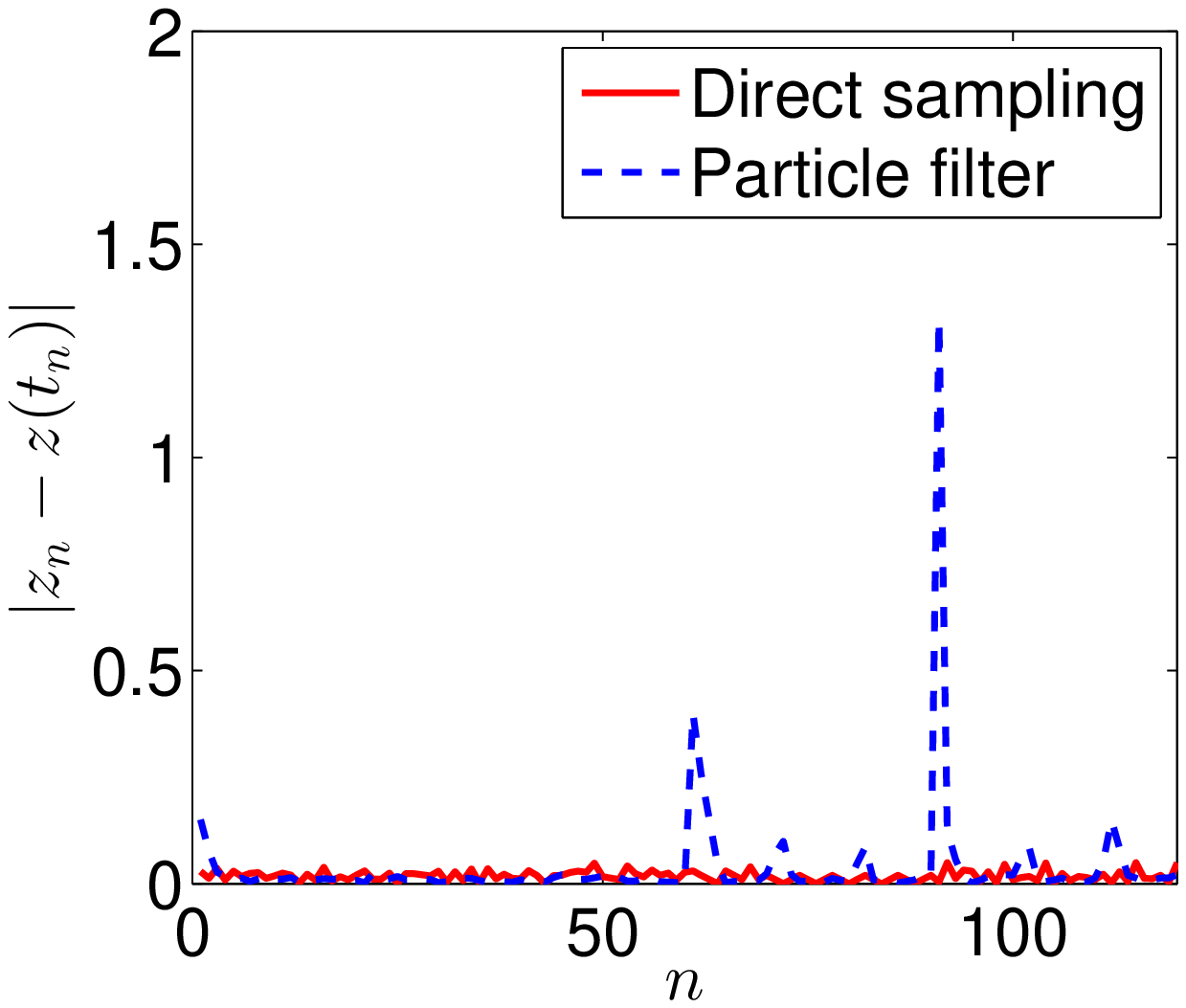}}\hfill
\caption{\label{fig:ex3} Reconstruct the trajectory of a text. (a) exact moving trajectory,
(b) reconstruction by  the direct sampling method, (c) reconstruction by the particle filter method with $N_s=500$, (d) point-wise error between the exact and reconstructed trajectory.}
\end{figure}

\begin{figure}
\hfill\subfigure[]{\includegraphics[width=0.45\textwidth]{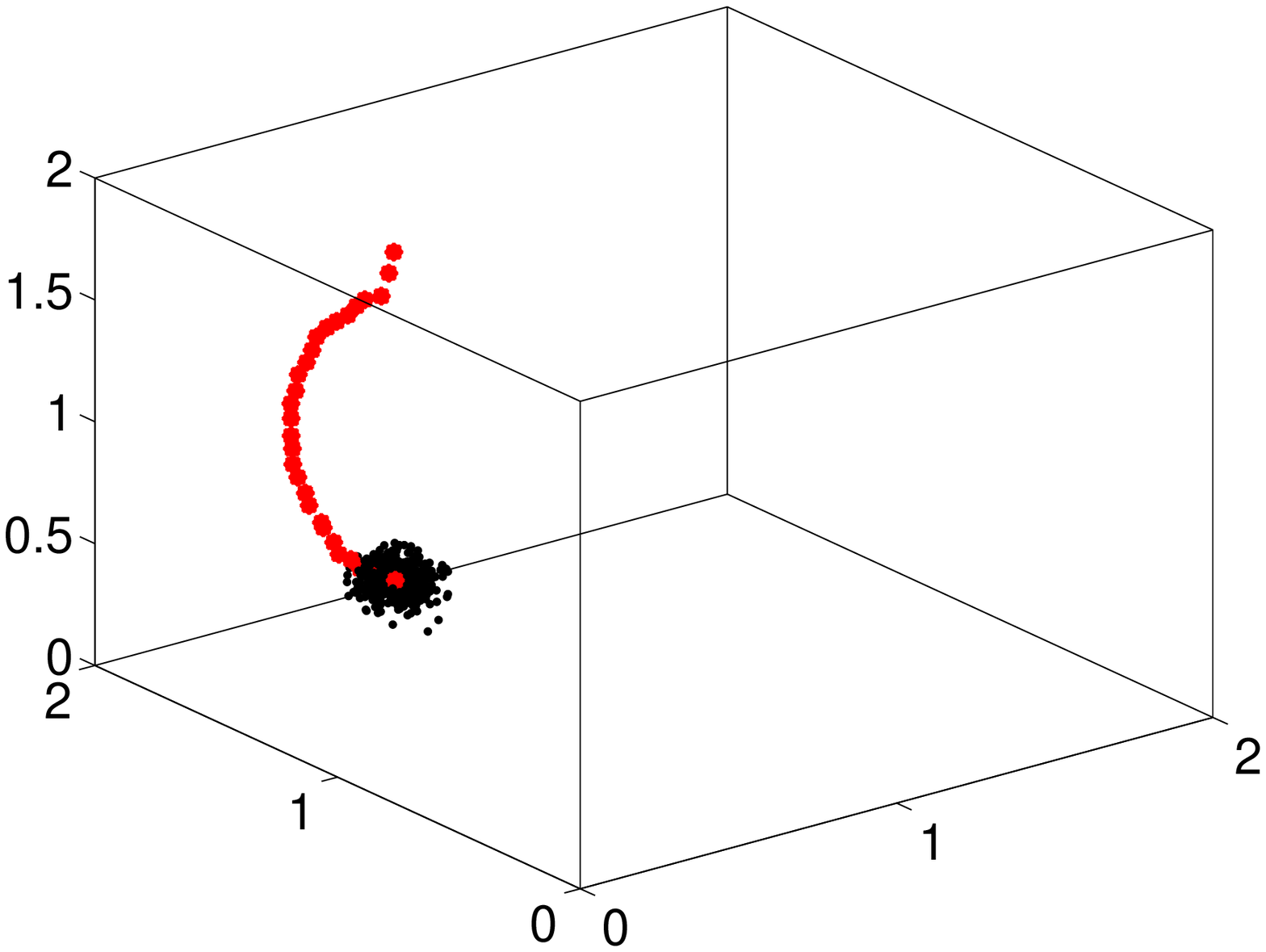}}\hfill
\hfill\subfigure[]{\includegraphics[width=0.45\textwidth]{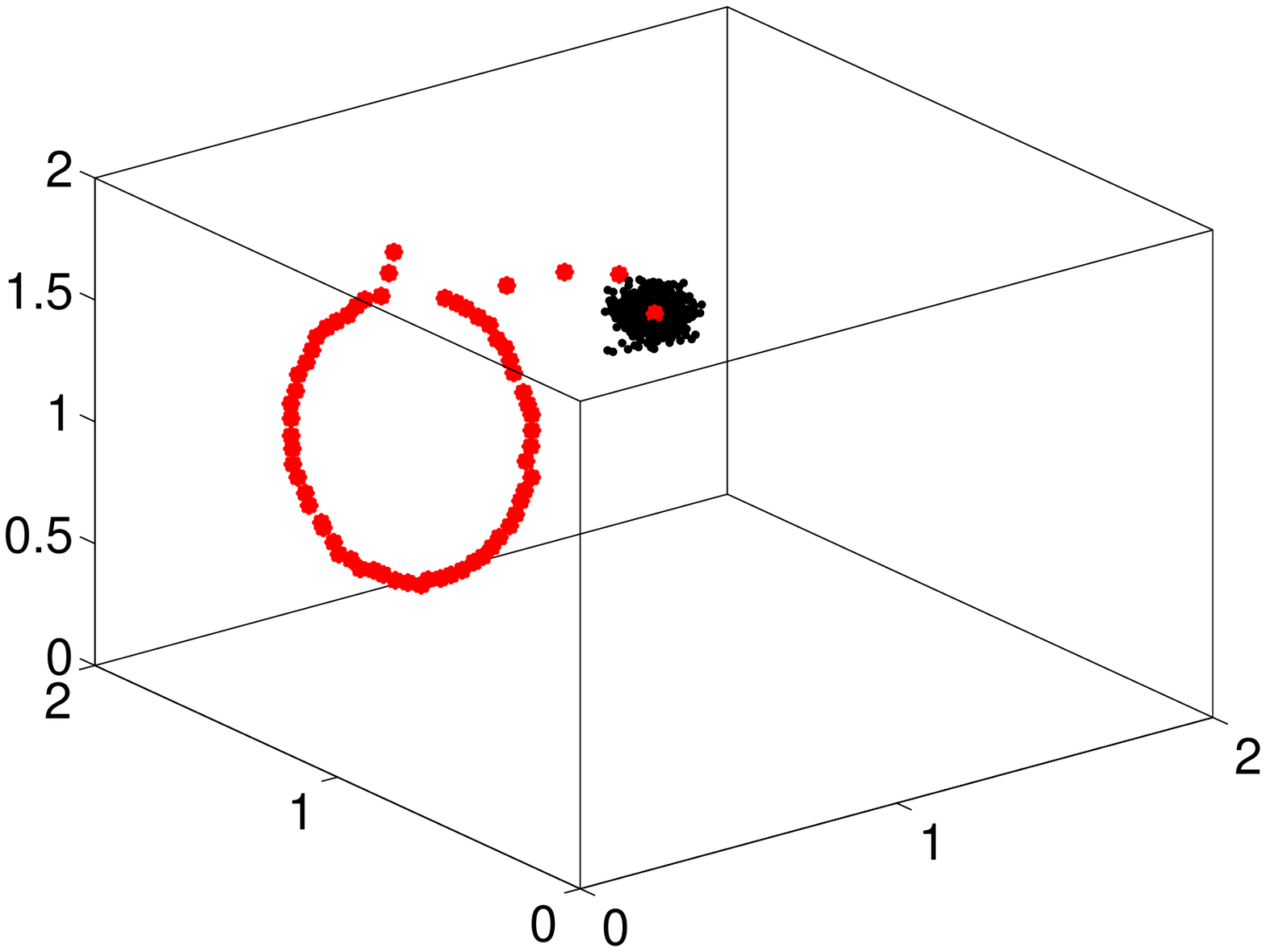}} \hfill\\
\hfill\subfigure[]{\includegraphics[width=0.45\textwidth]{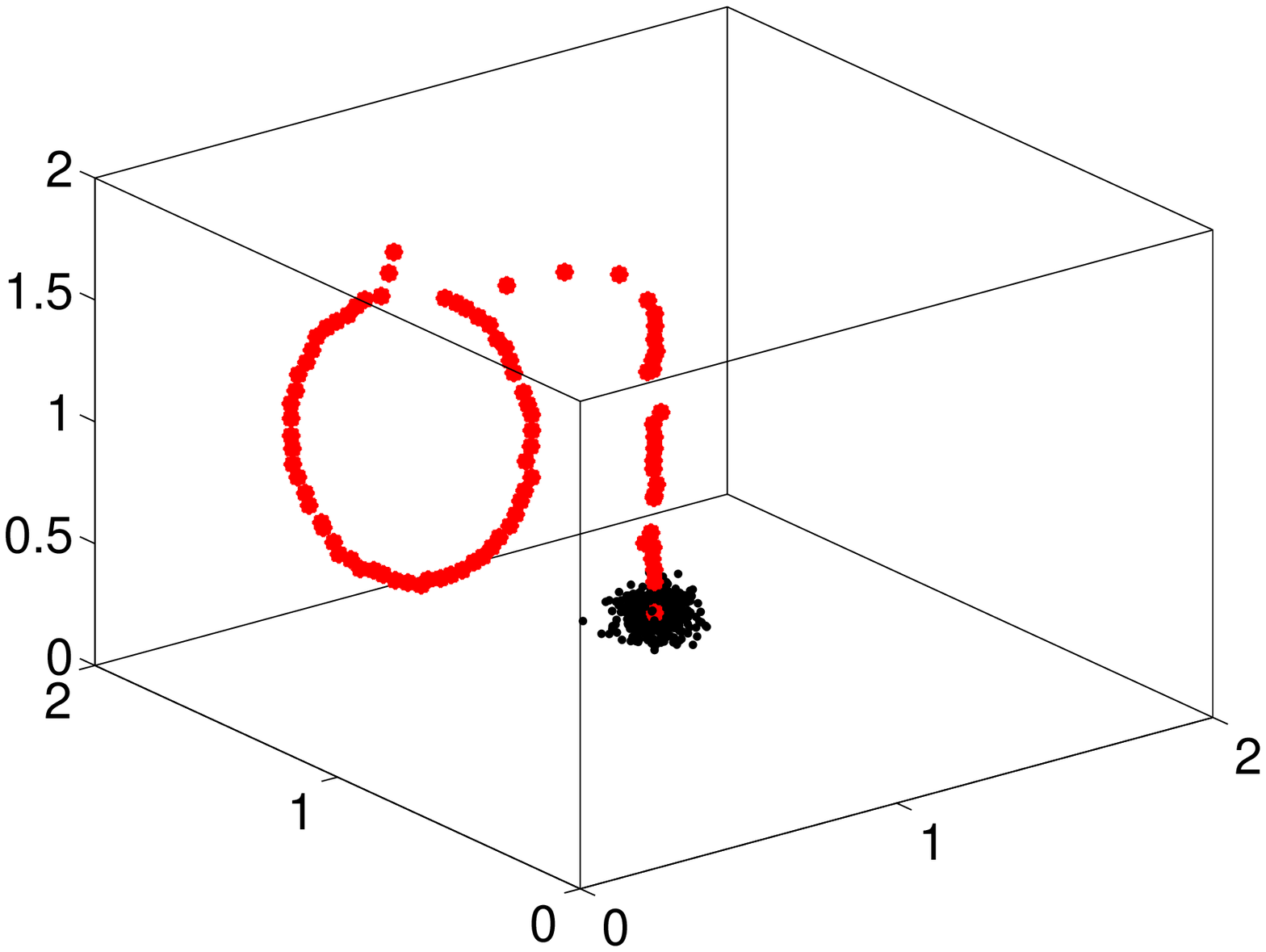}}\hfill
\hfill\subfigure[]{\includegraphics[width=0.45\textwidth]{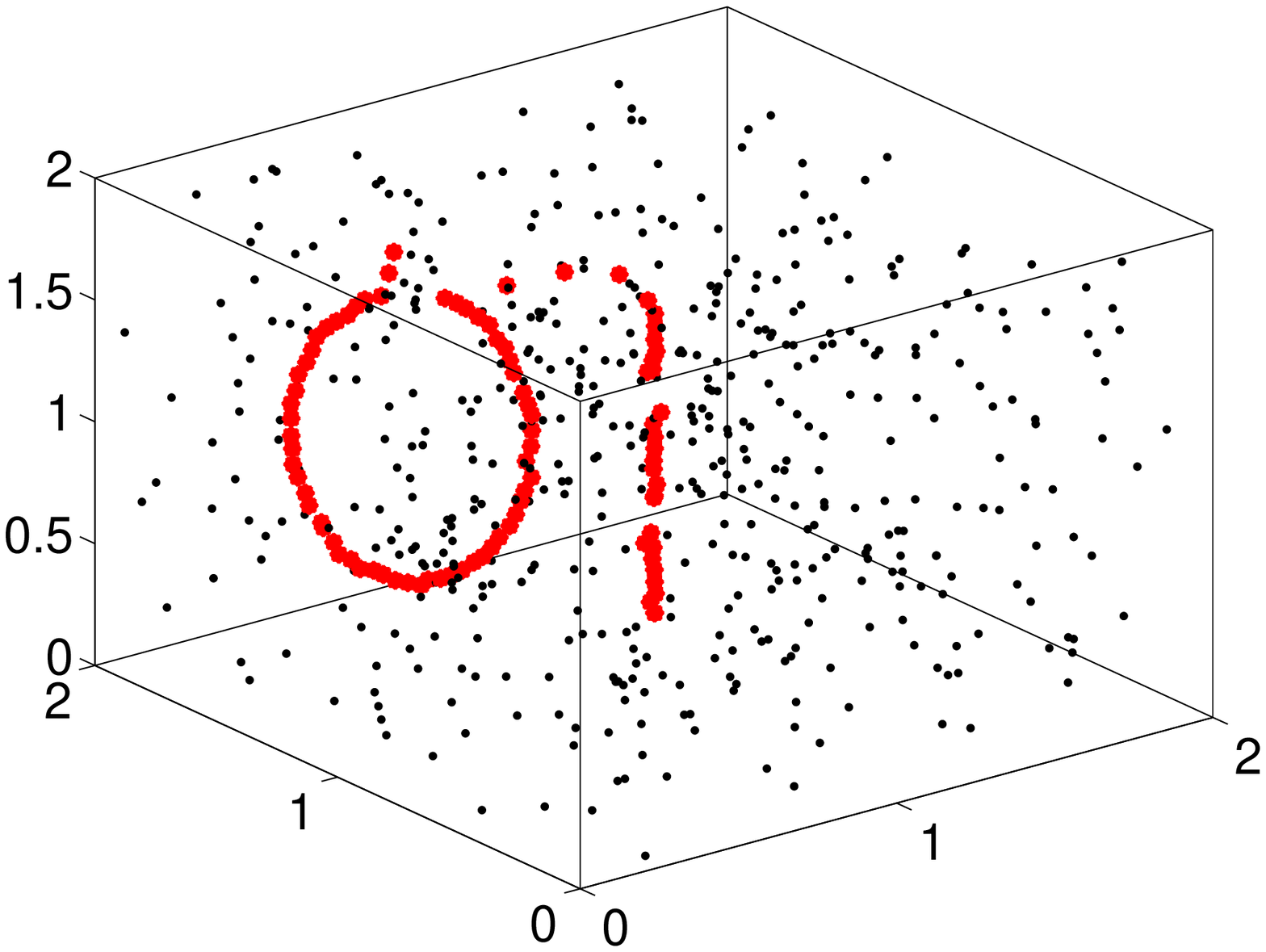}}\hfill\\
\hfill\subfigure[]{\includegraphics[width=0.45\textwidth]{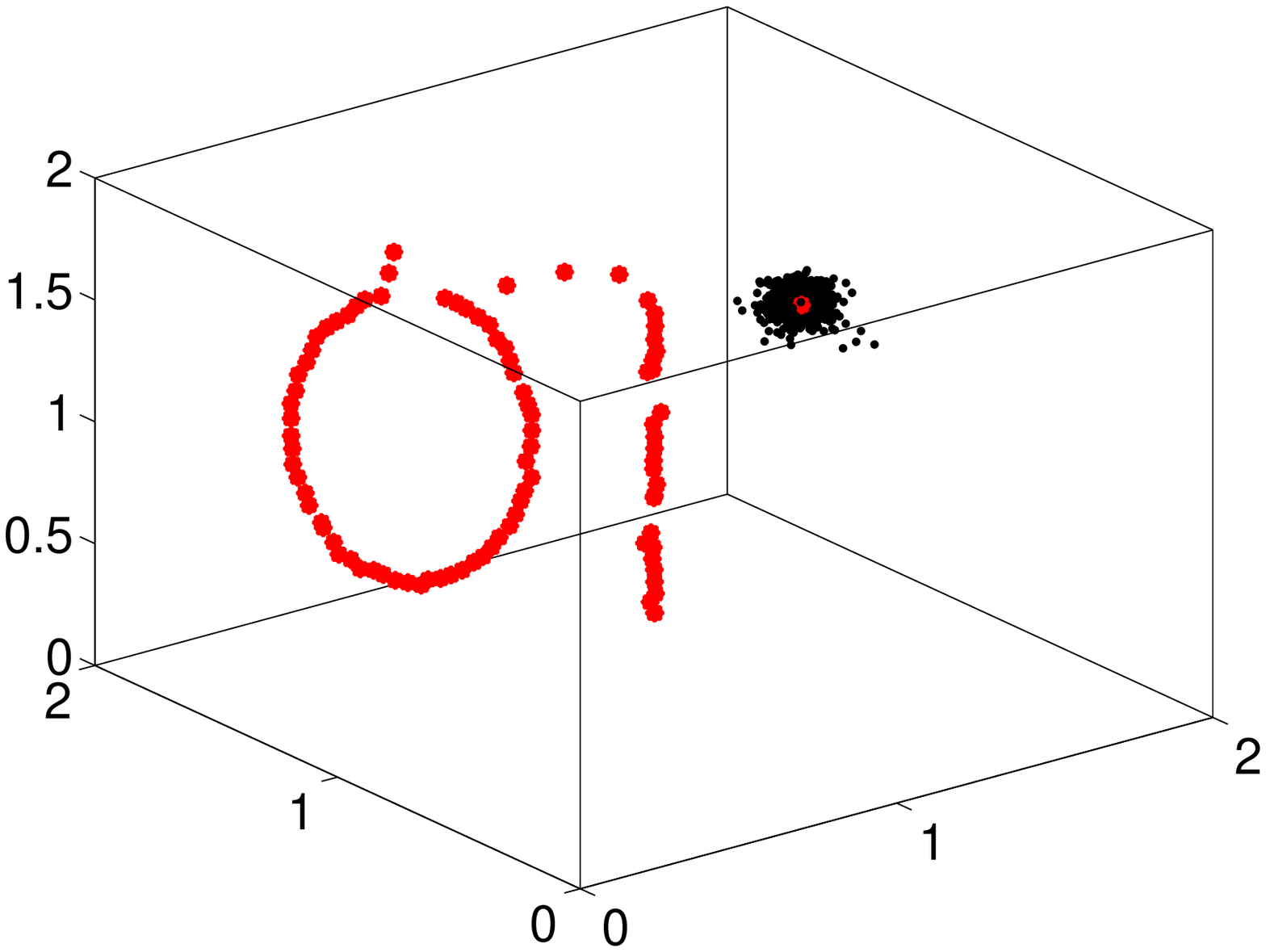}}\hfill
\hfill\subfigure[]{\includegraphics[width=0.45\textwidth]{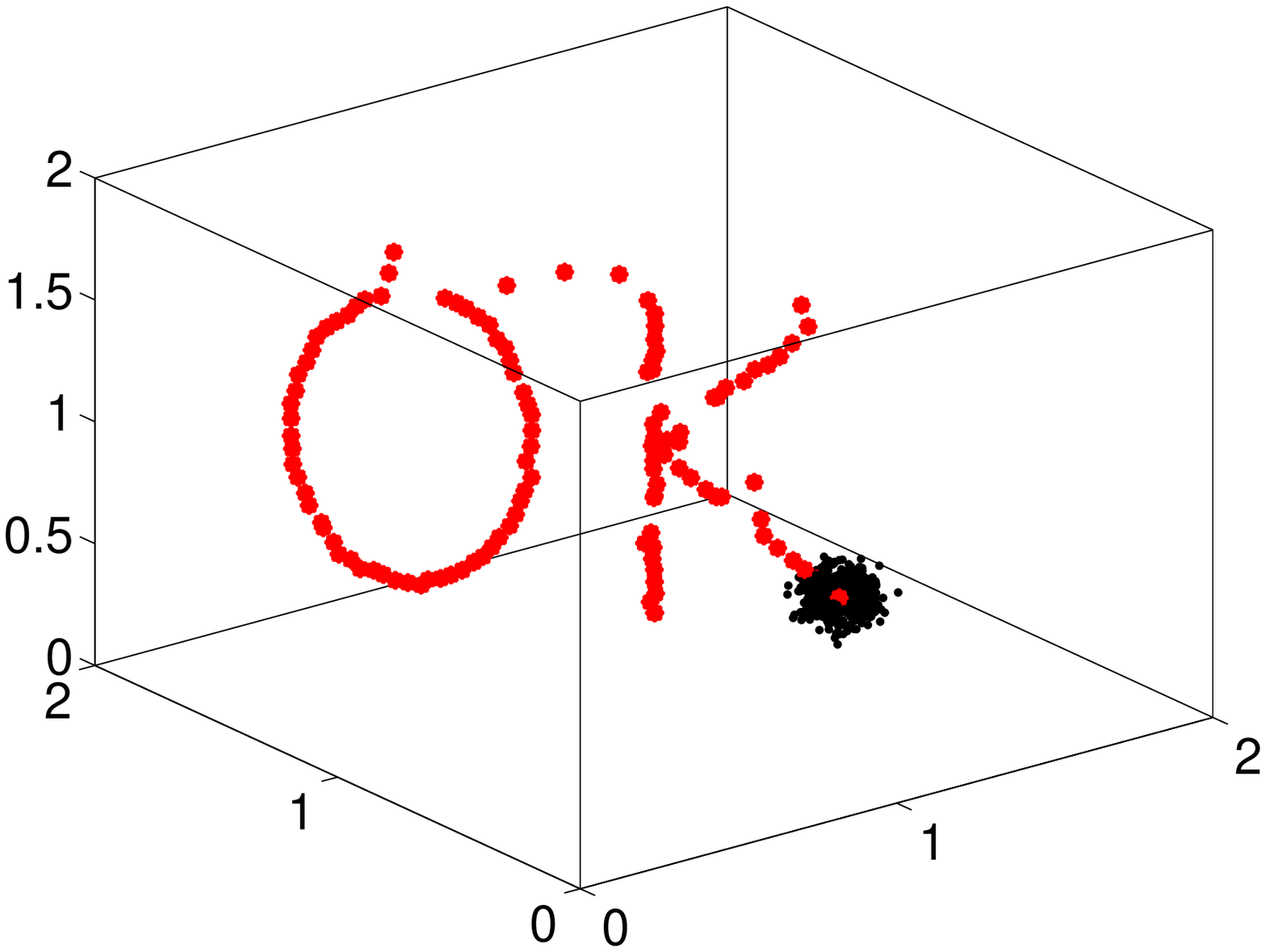}}\hfill
\caption{\label{fig:ex4} Reconstruct the moving trajectory by the particle filter method and snapshots at different instants, where the black points denotes the particles. (a) $t=3.0$ s, (b) $t=6.6$ s, (c) $t=9.0$ s, (d) $t=9.1$ s, (e) $t=9.2$ s, (f) $t=12.0$ s}
\end{figure}

This example is challenging since there exists a sudden skip in the moving trajectory, where the skip distance is far longer than the normal transition step size. The dotted  lines in Figure \ref{fig:ex3}(a) represent the skip trajectory. As discussed in Example \ref{exp1}, the standard deviation $\gamma$ usually  depends on normal transition step size. We note that the particle filter method  will stop  tracking the real trajectory path when there is a sudden skip. An efficient method  is to redistribute the particles such as {\bf Step 3} in {\bf Algorithm 2}, see Figure \ref{fig:ex4}(d). Figure \ref{fig:ex4} illustrates that the particle filter method could recover the trajectory  after redistribution. However, in Figure \ref{fig:ex3}(d), there exists large perturbation  at the skip instants, this is due to the shortcoming of the particle filter method.

\begin{rem}
In this paper, the reconstructed trajectories are only in the discrete form. To obtain a smooth curve as the reconstructed trajectory,  one could resort to the truncated Fourier approximation for post-processing the discrete points $\{z_n\}_{n=1}^{N_t}$\cite{GLLW}. The implementation of the direct sampling method could also be accelerated by incorporating the sequential or parallel local sampling strategies, see \cite{GLLW} for more details.
\end{rem}


\section{Conclusions}

In this work, we develop a conceptual framework of the novel gesture recognition techniques using electromagnetic waves. The model problem is formulated as an inverse source problem of determining the moving trajectory of a moving point charge. Two methods are proposed to deal with the inverse problem, respectively, in the deterministic and statistical manners. Mathematical justifications are presented and extensive numerical examples are provided to validate the effectiveness and efficiency of the methods.

We would like to remark that this work is mainly from a theoretical point of view and the physical and engineering realizations are beyond the scope of our current study.


\section*{Acknowledgments}

The work of Y. Guo was supported by the NSF grant of China (No.\,11601107, 11671111, 41474102). The work of J. Li was supported by the NSF grant of China  (No.\,11571161), the Shenzhen Sci-Tech Fund (No.\,JCYJ20160530184212170) and the SUSTech Startup fund.  The work of H. Liu was supported by the startup fund and FRG grants from Hong Kong Baptist University, Hong Kong RGC General Research Fund, No.\,12302415. The work of X. Wang was supported by the NSF grant of China (No.\, 11671113).


\appendix
\section{Appendix}\label{App}

In this appendix, for self-containedness and convenience of readers,  we collect some basic ingredients of the particle filter method. We also refer to \cite{ADP, CD, KS} and the references therein for more relevant details.

Define
 $\{X_t\}_{t=0}^{\infty}$ and $\{Y_t\}_{t=1}^{\infty}$ by  two stochastic processes, where $X_t\in \mathbb{R}^{n_x}$ represents  system state and  $Y_t\in \mathbb{R}^{n_y}$ represents the measurement at the $t$-th time instant.
 Let us consider the following time discrete dynamic model,
 \begin{align*}
 & X_{t}=f(X_{t-1})+V_{t},\\
 & Y_t=g(X_t)+W_t,
 \end{align*}
where $f:\mathbb{R}^{n_x}\mapsto \mathbb{R}^{n_x} $ and $g:\mathbb{R}^{n_y}\mapsto \mathbb{R}^{n_y}$ are  two  measurable functions.
 $V_t\in \mathbb{R}^{n_x}$ denotes the process noise and $W_t \in \mathbb{R}^{n_y}$ denotes the measurement noise.
 Let $p_V$ be a  density distribution of $V_t$ and $p_W$ be a  density distribution of $W_t$ and define
\begin{equation}\label{eq:pro1}
\begin{aligned}
 & p(x_t\mid  x_{t-1}):=p_V(x_t- f(x_{t-1})),\\
 & p(y_t\mid  x_t):=p_W(y_t- g(x_t)).
\end{aligned}
\end{equation}
 The optimal nonlinear filtering problem is to find the posterior probability density function $p(x_t\mid y_{1:t})$ at the state $x_t$ from the observation data $y_{1:t}:=\{y_1,\cdots,y_t\}$. Here  the posterior probability density function is given by Bayes' formula \cite{CD}:
   \begin{equation}\label{eq:PDF}
   p(x_t\mid y_{1:t})=\frac{p(y_t\mid x_t)p(x_t\mid y_{1:t-1})}{\int_{\mathbb{R}^{n_x}} p(y_t\mid x_t) p(x_t\mid y_{1:t-1})\, \mathrm{d} x_{t}},
 \end{equation}
 where
 \begin{equation*}
   p(x_t\mid y_{1:t-1})=\int_{\mathbb{R}^{n_x}} p(x_t\mid x_{t-1}) p(x_{t-1}\mid y_{1:t-1})\, \mathrm{d} x_{t-1}.
 \end{equation*}

The key idea behind the particle filter method is to use a set of samples with weight to approximate the posterior probability density function in \eqref{eq:PDF}.
Given $N$ random particles $\{x_t^{(i)} \}_{i=1}^N$, correspondingly,
the posterior probability density function  with  associated weights $\{w_t^{(i)} \}_{i=1}^N$ could be represented by
\begin{equation}\label{eq:mesure}
 p^N( x_t\mid y_{1:t}):=\frac{1}{\mathrm{d}x_t}\sum\limits_{i=1}^N w_t^{(i)} \delta_{x_t^{(i)}}(\mathrm{d}x_t),
\end{equation}
where $\delta_x$ denotes the delta-Dirac mass located at $x$ and
\begin{equation}\label{eq:weights}
   w_t^{(i)}=\frac{ w_{t-1}^{(i)}p(y_t \mid x_t^{(i)} )}{\sum\limits_{i=1}^N  w_{t-1}^{(i)} p(y_t\mid x_t^{(i)})}.
\end{equation}
Equation \eqref{eq:weights} indicates that it is possible that only one particle has a significant weight value after several recursive steps. A resampling stage \cite{CD} allows to replace the samples with low weights by copies of the samples with high weights, which ensures more particles in statistically significant areas.
The classical particle filter algorithm proceeds in three main steps,  see Algorithm PF.

\begin{table}[htp]
\centering
\begin{tabular}{cp{.8\textwidth}}
\toprule
\multicolumn{2}{l}{{\bf Algorithm PF:}\quad The classical particle filter method.}\\
\midrule
 {\bf Step 1} &  Initialization. For $t = 0$, sample
 \[
 x_0^{(i)}\sim p(x_0),\quad i=1,\cdots, N,
 \]
and set $t=1$. \\
{\bf Step 2} & Importance sampling. For $t\geq 1$, sample
\[
\tilde{x}_t^{(i)}\sim p(x_t\mid x_{t-1})\, p^N( x_{t-1}\mid y_{1:t-1}),\quad i=1,\cdots, N,
\]
and evaluate the normalized importance weights
\[
w_t^{(i)}= \frac{p(y_t\mid \tilde{x}_t^{(i)})}{\sum\limits_{i=1}^N  p(y_t\mid \tilde{x}_t^{(i)})}, \quad i=1,\cdots, N.
\]
\\
{\bf Step 3} & Resampling. Sample
\[
x_t^{(i)}\sim p^N(x_t\mid y_{1:t}),\quad i=1,\cdots, N,
\]
where the posterior probability density function is 
\[
p^N(x_t\mid y_{1:t})=\frac{1}{\mathrm{d}x_t}\sum\limits_{i=1}^N  w_t^{(i)} \delta_{\tilde{x}_t^{(i)}}(\mathrm{d}x_t), \quad i=1,\cdots, N.\]
\\
\bottomrule
\end{tabular}
\end{table}



\begin{thebibliography}{99}

\bibitem{1Mon} R. Albanese and P. Monk, {\it The inverse source problem for Maxwell's equations}, Inverse Problems, {\bf 22} (2006), 1023--1035.

\bibitem{2BaoG} H. Ammari, G. Bao and J. Fleming, {\it An inverse source problem for Maxwell's equations in magnetoencephalography}, SIAM J. Appl. Math., {\bf 62} (2002), 1369--1382.

\bibitem{Amm} H. Ammari, T. Boulier, J. Garnier, H. Kang, and H. Wang, {\it Tracking of a mobile target using generalized polarization tensor}, SIAM J. Imaging Sci., {\bf 6} (2013), 1477--1498.

\bibitem{ATW} H. Ammari, M. P. Tran and H. Wang, {\it Shape identification and classification in echolocation}, SIAM J. Imaging Sci., {\bf 7} (2014),  pp.~1883--1905.

 \bibitem{ADP} C. Andrieu, A. Doucet and E. Punskaya, {\it Sequential Monte Carlo Methods in Practice}, Springer-Verlag, New York,  2013, pp. 87--88.



\bibitem{CB} M. Cheney and B. Borden, {\it Imaging moving targets from scattered waves}, Inverse Problems, {\bf 24} (2008), 035005.


\bibitem{CK} {D.~Colton and R.~Kress}, {\it Inverse Acoustic and Electromagnetic Scattering Theory}, 3nd Edition, Springer-Verlag, New York, 2013.

\bibitem {CD} {D. Crisan and A. Doucet}, {\it A survey of convergence results on particle filtering methods for practitioners}, IEEE Trans. Signal Process, {\bf 50} (2002),  736--746.

 \bibitem{DLU} Y. Deng, H. Liu and G. Uhlmann, {\it On an inverse boundary problem arising in brain imaging}, arXiv:1702.00154

\bibitem{10Fokas} A. Fokas, Y. Kurylev and V. Marinakis, {\it The unique determination of neuronal currents in the brain via magnetoencephalography}, Inverse Problems, {\bf 20} (2004), 1067--1082.

\bibitem{GF} J. Garnier and M. Fink, {\it Super-resolution in time-reversal focusing on a moving source}, Wave Motion, {\bf 53} (2015), 80--93.

\bibitem{Griffiths} {D. Griffiths}, {\it Introduction to Electrodynamics}, 3nd Edition,
Prentice Hall, New Jersey, 1999, p. 438.

\bibitem{GLLW} Y. Guo, J. Li, H. Liu and X. Wang, {\it Mathematical design of a novel input/instruction device using a moving emitter}, arxiv.1609.05205,  (2016).

\bibitem{HR} S. He and V. Romanov, {\it Identification of dipole sources in a bounded domain for Maxwell's equations}, Wave Motion, {\bf 28} (1998),  25--40.



\bibitem{Isa2} V. Isakov, {\it Inverse Source Problems}, Mathematical Surveys and Monographs, 34. American Mathematical Society, Providence, 1990.

\bibitem{KS} {J. Kaipio and E. Somersalo}, {\it Statistical and Computatioanal Inverse Problems}, Springer-Verlag, New York, 2005.

\bibitem{KR} M. V. Klibanov and V. G. Romanov,  {\it Two reconstruction procedures for a 3D phaseless inverse scattering problem for the generalized Helmholtz equation}. Inverse Problems, {\bf 32} (2016), 015005.



\bibitem{Leis} R. Leis, {\it Initial Boundary Value Problems in Mathematical Physics}, Springer Fachmedien Wiesbaden, 1986.

\bibitem{LWY} {H. Liu, Y. Wang and C. Yang}, {\it Mathematical design of a novel gesture-based instruction/input device using wave detection}, SIAM J. Imaging Sci., {\bf 9} (2016), no. 2, 822--841.


\bibitem{Yam2} Y. Liu, D. Jiang and M. Yamamoto, {\it Inverse source problem for a double hyperbolic equation describing the three-dimensional time cone model}, SIAM J. Appl. Math., {\bf 75} (2015),  2610--2635.

\bibitem{24Naka} E. Nakaguchi, H. Inui and K. Ohnaka, {\it An algebraic reconstruction of a moving point source for a scalar wave equation}, Inverse Problems, {\bf 28} (2012), 065018.

\bibitem{Ned} J.-C. N\'ed\'elec, {Acoustic and Electromagnetic Equations: Integral Representations for Harmonic Problems}, Springer-Verlag, New York, 2001.

\bibitem{Potthast} {G. Nakamura and R. Potthast}, {\it Inverse Modeling}, IOP Publishing, Bristol, 2015.



 \bibitem{Wiki} Wikipedia, {\tt https://en.wikipedia.org/wiki/Gesture$_{-}$recognition}.





\end{thebibliography}
\end{document}